\documentclass[groupedaddress,10pt]{article}
\usepackage[utf8]{inputenc}
\usepackage{amsmath,amsthm,amsfonts,latexsym,amssymb,bm,enumerate}
\usepackage{graphicx}
\usepackage{ae}
\usepackage{cite}
\usepackage{float}
\usepackage{lmodern}
\usepackage[T1]{fontenc}
\usepackage{epsf}
\usepackage{psfrag}
\usepackage{mathrsfs}
\usepackage{authblk}
\usepackage{hyperref}
\DeclareGraphicsExtensions{.eps,.art,.ART,.ps}
\newcommand{\tr}{\operatorname{tr}}

\newcommand{\uk}{\underbar{$\kappa$}}
\newcommand{\rz}{r(0,0)}
\newcommand{\hu}{\tilde u}
\newcommand{\hv}{\tilde v}
\newcommand{\tu}{\tilde u}
\newcommand{\tv}{\tilde v}

\newcommand{\IVP}{${\mathrm{IVP}}_{\mathcal{C}^+}$}
\newtheorem{Thm}{Theorem}[section]
\newtheorem{proposition}[Thm]{Proposition}
\newtheorem{Prop}[Thm]{Proposition}
\newtheorem{lemma}[Thm]{Lemma}

\newtheorem{Cor}[Thm]{Corollary}

\newtheorem{Remark}[Thm]{Remark}

\newcounter{mnotecount}[section]

\renewcommand{\themnotecount}{\thesection.\arabic{mnotecount}}

\newcommand{\mnote}[1]
{\protect{\stepcounter{mnotecount}}$^{\mbox{~\footnotesize
$
\bullet$\themnotecount}}$ \marginpar{
\raggedright\tiny\em
$\!\!\!\!\!\!\,\bullet$\themnotecount: #1} }

\begin{document}
\title{Cosmic no-hair in spherically symmetric black hole spacetimes}
\author[1,3]{Jo\~ao L. Costa}
\author[2,3]{Jos\'e Nat\'ario}
\author[2,3]{Pedro Oliveira}
\affil[1]{Mathematics Department, Lisbon University Institute -- ISCTE}
\affil[2]{Mathematics Department, Instituto Superior T\'ecnico, ULisboa}
\affil[3]{CAMGSD, Instituto Superior T\'ecnico, ULisboa}
%

\maketitle
%
%
\begin{abstract}
We analyze in detail the geometry and dynamics of the cosmological region arising in spherically symmetric black hole solutions of the Einstein-Maxwell-scalar field system with a positive cosmological constant. More precisely, we solve, for such a system, a characteristic initial value problem with data emulating a dynamic cosmological horizon. Our assumptions are fairly weak, in that we only assume that the data approaches that of a subextremal Reissner-Nordstr\"om-de Sitter black hole, without imposing any rate of decay. We then show that the radius (of symmetry) blows up along any null ray parallel to the cosmological horizon (``near'' $i^+$), in such a way that $r=+\infty$ is, in an appropriate sense, a spacelike hypersurface.  We also prove a version of the Cosmic No-Hair Conjecture by showing that in the past of any causal curve reaching infinity both the metric and the Riemann curvature tensor asymptote to those of a de Sitter spacetime. Finally, we discuss conditions under which all the previous results can be globalized.
\end{abstract}
\tableofcontents
%
%
%
%
%
\section{Introduction}\label{sectionIntro}

Adding a positive cosmological constant to the Einstein field equations provides, arguably, the simplest way to model cosmological spacetimes undergoing accelerated expansion. Such expansion has remarkable effects in the global causal structure of solutions, as well as their asymptotic behavior. In fact, it is widely expected that, generically, for late cosmological times, the expansion will completely dominate the dynamics, damping all inhomogeneities and anisotropies, in such a way that, in the end, only the information contained in the cosmological constant will persist; given these heuristics, it is then natural to predict the asymptotic approach to a de Sitter solution, the simplest and most symmetric solution of the field equations with a positive cosmological constant. This expectation was, early on, substantiated in the celebrated Cosmic No-Hair Conjecture.

As usual in General Relativity, obtaining a precise statement for this conjecture is one its biggest challenges. Right from the start one has to deal with the fact that the Nariai solutions provides a simple counterexample for all the claims above; but since this solution is widely believed to be unstable~\cite{Beyer} (compare with~\cite{gajic}), cosmic no-hair is expected to hold generically. It turns out that the standard formulation of the conjecture does not differ much from the imprecise picture provided in the previous paragraph. It is in fact sufficiently vague so that a considerable amount of relevant work, concerning the global structure of solutions of the Einstein equations with a positive cosmological constant, immediately fits the picture; some notable examples are provided by the following (incomplete) list~\cite{Wald:1983,Friedrich:1986,Rendall:2004,Tchapnda:2003,Tchapnda:2005,Rodnianski:2009,lubbeKroon,Speck:2012,RingUniv,Oli:2016,Friedrich:2016,AH,Rad,dafRen,schlue2,CAN,Costa:2013}.

A natural, but hopeless, attempt to capture the idea of cosmic no-hair is to say that it holds provided that generic solutions contain foliations of a neighborhood of (future null) infinity ${\cal I}^{+}$ along which they approach de Sitter uniformly. Although tempting, it fails to capture some subtleties created by {\em cosmic silence}. To better understand these let us for a moment consider the linear wave equation, in a fixed de Sitter~\footnote{Or in Reissner-Nordstr\"om-de Sitter~\cite{JPZ, schlue} or Kerr-de Sitter~\cite{schlue}.} background, as a proxy to the Einstein equations. In this analogy, the previous formulation of cosmic no-hair would require (generic) solutions to decay to a constant at ${\cal I}^{+}$. However, it is well known that this is not the case~\cite{Rendall:2004,JPZ,schlue}.
To get an idea of why, note that cosmic silence is illustrated by the following construction: given any two inextendible causal curves $\gamma_i$, $i=1,2$, parameterized by proper time $t$ and reaching ${\cal I}^{+}$ as $t\rightarrow\infty$, there is ``a late enough'' ($t\gg 1$) Cauchy surface $\Sigma_t$ such that the sets ${\cal D}_{i}=J^{-}(\gamma_i)\cap J^{+}(\Sigma_t)$ are disjoint. Using the fact that the wave equation admits constant solutions, we prescribe a different constant at each ${\cal D}_i$, and then construct initial data on $\Sigma_t$ by gluing the data induced on $J^{-}(\gamma_i)\cap \Sigma_t$ by the fixed constants; finally, we solve the wave equation for such data. In the end we obtain a solution which is not constant on ${\cal I}^{+}$! This shows that the previous ``uniform approach'' to de Sitter is too strong a requirement.  A similar conclusion can be obtained for the full Einstein equations by considering the freedom we have to prescribe data at ${\cal I}^{+}$ for the Friedrich conformal field equations~\cite{Friedrich:1986,Friedrich:2016,kroonBook}.

A precise formulation of cosmic no-hair, properly taking into account the previous issues, was obtained only recently by Andr\'{e}asson and Ringstr\"om~\cite{AH}. The fundamental new insight is to relax the previous requirement and instead say that cosmic no-hair holds if for each  future inextendible causal curve $\gamma$ and ${\cal D}_t=J^{-}(\gamma)\cap J^{+}(\Sigma_t)$, defined  as before, $({\cal D}_t,g)$ approaches (in a precise sense~\footnote{See point 4 of Theorem~\ref{thmMain} for the formulation used here.}) de Sitter as $t\rightarrow \infty$. In other words, every such observer will see the spacetime structure around him approaching that of a de Sitter spacetime, although, in general, this will not happen in a uniform way for all observers. Here we will follow the spirit, although not the letter, of
the Andr\'{e}asson-Ringstr\"om formulation. A distinct feature, for instance, will come from the fact that we will be considering black hole spacetimes, in which not all future inextendible causal curves reach infinity. Clearly, the corresponding observers will not see the geometry approaching that of de Sitter, and so the Cosmic No-Hair Conjecture will have to be suitably modified to exclude them.


It was also shown in~\cite{AH} that this form of cosmic no-hair does hold generically for the Einstein-Vlasov system under suitable toroidal-symmetric assumptions (see also~\cite{Rad} for results concerning scalar fields in similar symmetric settings). We expect that this form of cosmic no-hair should be valid for most, if not all, of the models considered in the references given above.

In the context of cosmic no-hair, either by symmetry assumptions or smallness conditions, all the existing (non-linear) results exclude the existence of black holes a priori. The main goal of this paper is to help bridge this gap. Nonetheless, there exist some notable partial exceptions:~\cite{dafRen} deals with the Einstein-Vlasov system under various symmetry assumptions and provides a very general, mostly qualitative, description of the global structure of the corresponding solutions. In particular, for spherically symmetric solutions, it is shown there that the Penrose diagram of the {\em cosmological region}~\footnote{See below the clarification of this terminology.}, if non-empty, is bounded to the future by an acausal curve where the radius of symmetry is infinite (assuming an appropriate non-extremality condition); this is in line with the expectations of cosmic no-hair but does not provide enough quantitative information to show that the geometry is approaching that of de Sitter. We also mention the recent work by Schlue\cite{schlue2}, where he takes a remarkable step towards the proof of the (non-linear) stability of the cosmological region of Schwarzschild-de Sitter for the (vacuum) Einstein equations, without any symmetry assumptions.

The aim of this paper is to provide the first (to the best of our knowledge) complete realization of cosmic no-hair, in the spirit of the Andr\'{e}asson-Ringstr\"om formulation, in the context of subextremal black hole spacetimes. With this aim in mind we take as reference solutions the subextremal elements of the Reissner-Nordstr\"om-de Sitter (RNdS) family~\footnote{To set the convention, whenever we refer to an RNdS solution $({\cal M}, g)$ we mean the maximal domain of dependence $D(\Sigma) = {\cal M}$ of a complete Cauchy hypersurface $\Sigma= \mathbb{R}\times \mathbb{S}^2$.} (see Section~\ref{sectionRNdS} for more details). These  model static, spherically symmetric and electromagnetically charged black holes in an expanding universe. They are solutions of the Einstein-Maxwell equations parameterized by mass $M$, charge $e$ and cosmological constant $\Lambda$. The global causal structure of these black hole spacetimes is completely captured by the Penrose diagram in Figure~\ref{figPenrose} and is considerably different from the one of their asymptotically flat ($\Lambda=0$) counterparts. The first noticeable difference corresponds to the existence of a periodic string of causally disconnected isometric regions. Among these regions, the only ones with no analogue in the $\Lambda=0$ case are the cosmological regions~\footnote{It is also standard to refer to these as expanding regions, but we prefer the designation ``cosmological'', since the local regions (regions I and III in Figure~\ref{figPenrose}) are also expanding~\cite{brillRNdS}.} (regions V and VI in Figure~\ref{figPenrose}), which, in the topology of the plane, are bounded by cosmological horizons $\cal C^{\pm}$, future and past null infinity ${\cal I}^{\pm}$, and a collection of points $i^{\pm}$. Since all of its connected components are isometric, we can focus on only one of them; for convenience, we will consider a future component, i.e., one whose boundary contains a component of ${\cal I}^{+}$ (region V in Figure~\ref{figPenrose}).
{\begin{figure}[ht]
\centering
\includegraphics[width=1.1\textwidth]{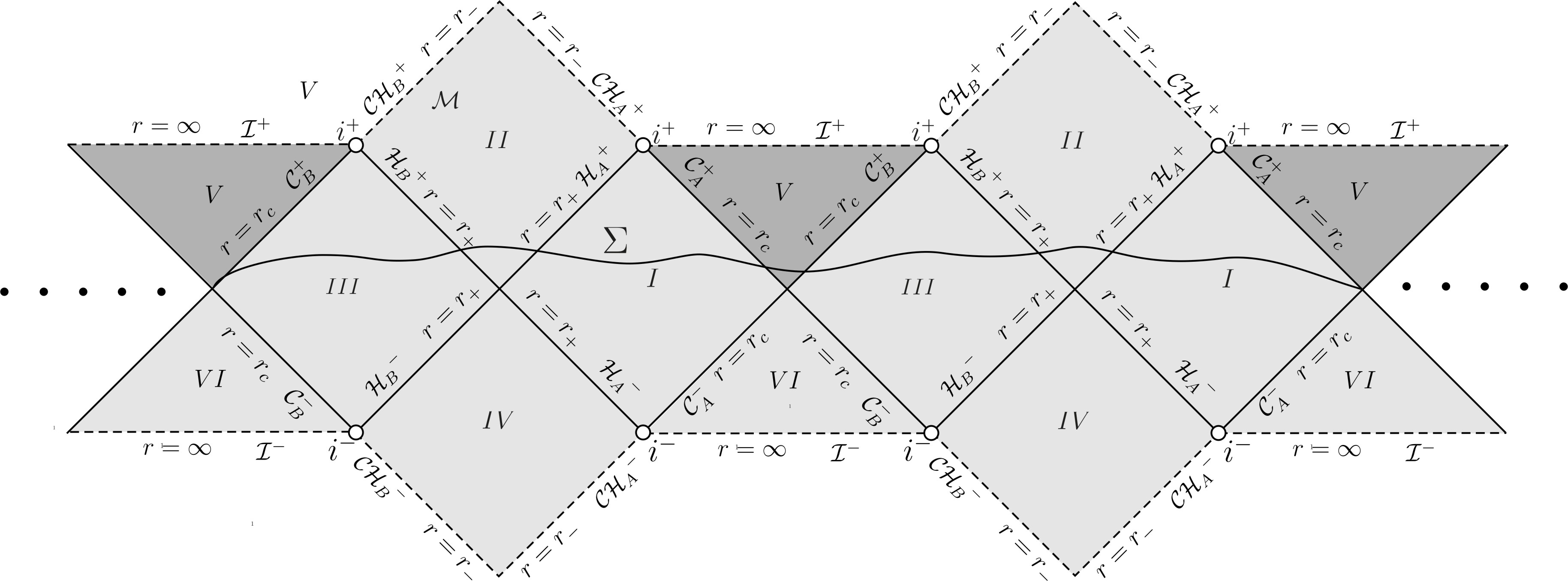}
\caption[Text der im Bilderverzeichnis auftaucht]{Conformal diagram of Reissner-Nordstr\"om-de Sitter spacetime.}
\label{figPenrose}\end{figure}}
%

We will start our evolution from a dynamical cosmological horizon, whose properties generalize those of the future Killing cosmological horizon ${\cal C}_{A}^{+}$ of our reference RNdS solutions; by imposing appropriate data on a transverse null surface to the cosmological horizon we then obtain a well-posed characteristic initial value problem.
The remarkable results by Hintz and Vasy~\cite{hintzVasy, Hintz}, concerning the non-linear stability of the local region (regions I and III in Figure~\ref{figPenrose}) of de Sitter black holes, suggest\footnote{These stability results concern the Einstein vacuum and Einstein–Maxwell equations, and not the Einstein–Maxwell–scalar field system.} that our defining properties for a dynamical horizon can, in principle, be recovered from (spherically symmetric) Cauchy initial data which is close, in an appropriate sense, to RNdS data. In fact, we expect our assumptions to hold for larger classes of Cauchy data, especially since although we will require our characteristic data to approach the data of a reference RNdS black hole, we will not need to impose any specific rates of decay.

In this paper, we will show that the future evolution of the characteristic initial data above will lead to a solution with the following properties: the radius (of symmetry) blows up along any null ray parallel to the cosmological horizon (``near'' $i^+$). Moreover, $r=+\infty$ is, in an appropriate sense, a (differentiable) spacelike hypersurface.  Finally we will also obtain a version of cosmic no-hair by showing that in the past of any causal curve reaching $r=\infty$, both the metric and the Riemann curvature tensor asymptote those of a de Sitter spacetime, at a specific rate. We will also briefly discuss conditions under which all the previous results can be globalized to a ``complete'' cosmological region.

\subsection{Main results}

Our main results can be summarized in the following:

\begin{Thm}
\label{thmMain}
Fix, as a reference solution, a subextremal member of the RNdS family, with parameters $(M,e,\Lambda)$ and cosmological radius $r_c$.

Let $({\cal M},g,F,\phi)$ be the maximal globally hyperbolic development of smooth spherically symmetric initial data given on $\Sigma= \mathbb{R}\times \mathbb{S}^2$, with charge parameter $e$ and cosmological constant $\Lambda$. Let $\tilde{\cal Q}$ be the projection of ${\cal M}$ to the (spherical symmetry) orbit space and assume that $\tilde {\cal Q}$ contains a future complete null line ${\cal C}^+$. Let
$(u,v):\tilde {\cal Q}\rightarrow\mathbb{R}^2$ be such that $u$ is an affine parameter on ${\cal C}^+=\{v=0\}$, $\partial_u$ and $\partial_v$ are both future oriented and, in an open set of $\tilde {\cal Q}$ containing ${\cal C}^+$,
$$g = - \Omega^2 (u,v)  du  dv + r^2 (u,v)  \mathring{g}\;,$$
where $\mathring{g}$ is the metric of the round 2-sphere.

Assume also that
\begin{enumerate}[(i)]
\item the radius of symmetry satisfies $r(u,0)\rightarrow r_c$, as $u\rightarrow \infty$\;,
\item the (renormalized) Hawking mass satisfies $\varpi(u,0)\rightarrow M$, as $u\rightarrow \infty$\;,
\item $\partial_vr(u,0)> 0$, for all $u\geq 0$\;, and
\item $|\partial_u\phi(u,0)|\leq C$, for some $C>0$ and all $u\geq 0$\;.
\end{enumerate}

Then, there exist $U,V>0$ such that, in ${\cal Q}=\tilde {\cal Q}\cap [U,\infty)\times[0,V)$, we have:
\begin{enumerate}
\item \underline{Blow up of the radius function:} for all $v_1\in(0,V)$, there exists $u^*(v_1)<\infty$ such that $r(u,v_1)\rightarrow \infty$, as $u\rightarrow u^*(v_1)$. Moreover $r(u,v)\rightarrow \infty$, as $(u,v)\rightarrow (u^*(v_1),v_1)$.
\item \underline{Asymptotic behavior of the scalar field:} there exists $C>0$ such that
$$
\left|\frac{\partial_v\phi}{\partial_vr}\right|+ \left|\frac{\partial_u\phi}{\partial_ur}\right|\leq C r^{-2}\;.
$$
\item \underline{$r=\infty$ is spacelike:} for a parameterization of the curves of constant $r$ of the form $v\mapsto (u_r(v),v)$ there exists a constant $C_1>0$ and a function $(0,V]\ni v\mapsto C_2(v)\in\mathbb{R}^+$, which may blow up as $v\rightarrow 0$, such that
$$-C_2(v)<u'_r(v)<-C_1\;,$$
holds for all $r>r_c$ and all $v>0$.
\item \underline{Cosmic no-hair:} Let $\gamma$ be a causal curve along which $r$ is unbounded. Let $i_{dS}$ be a point in the future null infinity of the de Sitter spacetime with cosmological constant $\Lambda$. Let $\{e_{I}\}_{I=0,1,2,3}$ be an orthonormal frame in de Sitter defined on $J^{-}(i_{dS})\cap \{r_{dS}>r_1\}$, for some $r_1>0$; in particular, in this frame, the de Sitter metric reads $^{dS}\!g_{IJ}=\eta_{IJ}$. Then, by increasing $r_1$ if necessary, there exists a diffeomorphism mapping $J^{-}(\gamma)\cap \{r>r_1\}$ to a neighborhood of $i_{dS}$ in $J^{-}(i_{dS})$ such that, in the dS frame $\{e_I\}$, we have, for $r_2 \geq r_1$,
\begin{equation}
\sup_{J^-(\gamma)\cap\{r\geq r_2\}}|g_{IJ}-^{dS}\!\!g_{IJ}|\lesssim r_2^{-2}\;,
\end{equation}
and
\begin{equation}
\sup_{J^-(\gamma)\cap\{r\geq r_2\}}|R^{I}_{JKL}-^{dS}\!\!R^{I}_{JKL}|\lesssim r_2^{-2}\;.
\end{equation}
\end{enumerate}
\end{Thm}


The proof of this theorem will be given in Section~\ref{sectionProofs}. Condition \emph{(iii)} is discussed further in Remark~\ref{confusao}.

We will now provide extra conditions that allow the previous results to be globalized to a
``complete'' cosmological region:

\begin{Thm}
\label{thmMainGlobal}
Fix, as reference solutions, two subextremal members of the RNdS family, with parameters $(M_i,e,\Lambda)$,  and cosmological radius $r_{c,i}$, $i=1,2$.
Let $({\cal M},g,F,\phi)$ and $(u,v)$ be as in Theorem~\ref{thmMain}, but now assume that $u=0$ is also complete and that
\begin{enumerate}[(i)']
\item  $r(x,0)\rightarrow r_{c,1}$ and $r(0,x)\rightarrow r_{c,2}$, as $x\rightarrow \infty$\;,
\item $\varpi(x,0)\rightarrow M_1$ and $\varpi(0,x)\rightarrow M_2$ as $x\rightarrow \infty$\;,
\item $\partial_vr(x,0)>0$ and $\partial_ur(0,x) > 0$, for all $x\geq 0$\;, and
\item $|\partial_u\phi(x,0)|+|\partial_v\phi(0,x)|\leq C$, for some $C>0$ and all $x\geq 0$\;.
\end{enumerate}
If, moreover, either
\begin{enumerate}[I.]
\item there exists a sufficiently small $\epsilon>0$, such that  for all $x\geq 0$
$$|r(x,0)-r_{c,1}|+|r(0,x)-r_{c,2}|+|\varpi(x,0)-M_1|+|\varpi(0,x)-M_2|\leq \epsilon\;,$$
or
\item we assume a priori that the future boundary of $\tilde {\cal Q}\cap [0,\infty)^2$, in $\mathbb{R}^2$, is of the form
$${\cal B}=\{(\infty,0)\}\cup{\cal N}_1\cup{\cal B}_{\infty}\cup {\cal N}_2\cup \{(0,\infty)\}\;,$$
where, for some $u_{\infty},v_{\infty}\geq 0$,  ${\cal N}_1=\{u=\infty\}\times (0,v_{\infty}]$ and ${\cal N}_2=(0,u_{\infty}]\times\{v=\infty\}$ are (possibly empty) null segments, and ${\cal B}_{\infty}$ is an acausal curve,  connecting $(\infty,v_{\infty})$ to $(u_{\infty},\infty)$, along which the radius $r$ extends to infinity by continuity,
\end{enumerate}
then the conclusions 1--4 of Theorem~\ref{thmMain} hold in ${\cal Q}=\tilde {\cal Q}\cap [0,\infty)^2$; in particular this implies that under assumption II the segments ${\cal N}_i$ are empty.
\end{Thm}

The proof of this theorem will be presented in Section~\ref{sectionProofs}.

\begin{Remark}
The a priori assumptions concerning the global structure of the future boundary of the Penrose diagram, formulated in assumption {\em II} above,  were inspired by the results in~\cite{dafRen}[Theorem 1.3] concerning the Einstein-Vlasov system.
Presumably, some of those results will remain valid for the matter model considered here, since they are mostly based on general properties, namely energy conditions and extension criteria, that should hold for both systems. Here we will not pursue this relation any further.
\end{Remark}

\section{Setting}\label{sectionSeetings}
\subsection{The Einstein-Maxwell-scalar field system in spherical symmetry}

We will consider the Einstein-Maxwell-scalar field (EMS) system in the presence of a positive cosmological constant $\Lambda$:
\begin{align}
& R_{\mu\nu} - \frac{1}{2} R\, g_{\mu\nu} + \Lambda g_{\mu\nu} = 2 T_{\mu\nu} \label{einstein} \, , \\
& T_{\mu\nu} = \partial_\mu \phi\, \partial_\nu \phi - \frac{1}{2} \partial_\alpha \phi\, \partial^\alpha \phi \, g_{\mu\nu} + F_{\mu\alpha} {F^\alpha}_\nu - \frac{1}{4} F_{\alpha\beta} F^{\alpha\beta} g_{\mu\nu} \,,
\label{energymomentumtensor} \\
& dF = d\star F = 0 \label{maxwell} \, , \\
& \square_g \phi = 0 \label{scalarfield} \, .
\end{align}
In this system a Lorentzian metric $g_{\mu\nu}$  is coupled to a Maxwell field $F_{\mu\nu}$ and a scalar field $\phi$ via the Einstein field equations~\eqref{einstein}, with energy-momentum tensor~\eqref{energymomentumtensor}. To close the system we also impose the  source-free Maxwell equations~\eqref{maxwell}. Then, the scalar field satisfies, as a consequence of the previous equations, the wave equation~\eqref{scalarfield}. Note that $\phi$ and $F_{\mu\nu}$ only interact through the gravitational field equations; consequently, the existence of a non-vanishing electromagnetic field requires a non-trivial topology for our spacetime manifold $\cal M$.

We will work in spherical symmetry by assuming that ${\cal M}={\cal Q}\times \mathbb{S}^2$ and by requiring the existence of double-null coordinates $(u,v)$, on $\cal Q$, such that the spacetime metric takes the form
\begin{equation}
g = - \Omega^2 (u,v)  du  dv + r^2 (u,v)  \mathring{g}\;,
\end{equation}
with $\mathring{g}$ being the metric of the unit round sphere.

Under these conditions, Maxwell's equations decouple from the EMS system. More precisely, there exists a constant $e\in\mathbb{R}$, to which we will refer as the charge (parameter), such that
\begin{equation}
\label{F}
F = - \frac{e \, \Omega^2(u,v)}{2 \, r^2(u,v)} \, du \wedge dv\;.
\end{equation}
\noindent
The remaining equations of the EMS system then become
\begin{align}
& \partial_u\partial_vr = -\frac{\Omega^2}{4r} - \frac{\partial_ur\,\partial_vr}{r} + \frac{\Omega^2e^2}{4r^3} + \frac{\Omega^2 \Lambda r}{4} \label{wave_r} \, , \\
& \partial_u\partial_v\phi = -\,\frac{\partial_ur\,\partial_v\phi+\partial_vr\,\partial_u\phi}{r} \label{wave_phi} \, , \\
& \partial_v\partial_u\ln\Omega = -\partial_u\phi\,\partial_v\phi-\,\frac{\Omega^2e^2}{2r^4}+\frac{\Omega^2}{4r^2}+\frac{\partial_ur\,\partial_vr}{r^2} \label{wave_Omega} \, , \\
& \partial_u \left(\Omega^{-2}\partial_u r\right) = -r\Omega^{-2}\left(\partial_u\phi\right)^2 \label{raychaudhuri_u} \, , \\
& \partial_v \left(\Omega^{-2}\partial_v r\right) = -r\Omega^{-2}\left(\partial_v\phi\right)^2\label{raychaudhuri_v} \, .
\end{align}

A well known consequence of the Raychaudhuri equations~\eqref{raychaudhuri_u} and \eqref{raychaudhuri_v} is the following version of Hawking's area law:

\begin{Prop} \label{area_Hawking}
If a radial null geodesic $\gamma$ is future complete and $r\circ\gamma$ does not vanish then $r\circ\gamma$ must be nondecreasing towards the future.
\end{Prop}

\begin{proof}
Since $r\circ\gamma$ does not vanish, we can assume that $\gamma$ is given by $v=0$. By rescaling the coordinate $u$, we can choose $\Omega^2(u,0)=1$, so that $u$ is an affine parameter (compare with~\cite{CGNS2}[Section 8]). In this case, the Raychaudhuri equation~\eqref{raychaudhuri_u} implies $\partial^2_u r \leq 0$, and the result then follows from the fact that $r\circ\gamma$ does not vanish and $u$ is unbounded by the future completeness of $\gamma$.
\end{proof}

It turns out to be convenient, both by its physical relevance and by its good monotonicity properties, to introduce the (renormalized) Hawking mass $\varpi=\varpi(u,v)$, defined by
\begin{equation}
\label{1-mu}
1-\mu:= 1-\frac{2\varpi}{r}+\frac{e^2}{r^2}-\frac{\Lambda}{3}r^2=\partial^{\alpha}r\partial_{\alpha} r\;.
\end{equation}
Later we will replace the wave equation for the conformal factor $\Omega$~\eqref{wave_Omega} by equations~\eqref{EvMass_u} and~\eqref{EvMass_v} prescribing the gradient of $\varpi$.

In this paper we will use a first order formulation of the EMS system, obtained by introducing the quantities:
\begin{align}
& \nu \, := \, \partial_u r \, ,\\
& \lambda \, := \, \partial_v r \, ,\\
& \theta \, := \, r \partial_v \phi \, ,\\
& \zeta \, := \, r \partial_u \phi \, .
\end{align}
Note, for instance, that in terms of these new quantities~\eqref{1-mu} gives
\begin{equation}
\label{1-mu2}
1-\mu=-4\Omega^{-2} \lambda\nu\,.
\end{equation}
Later we will be interested in a region where $\lambda>0$; there, we are allowed to define
\begin{equation}
\label{bark}
\underbar{$\kappa$}:=-\frac{1}{4}\Omega^{2}\lambda^{-1}\;.
\end{equation}
%
%


Our first order system is then\footnote{Note that $(1-\mu)$ depends on $(u,v)$ only through $(r,\varpi)$. In what follows, in a slight abuse of notation, we will interchangeably regard $(1-\mu)$ as a function of either pair of variables, with the meaning being clear from the context.}
\begin{align}
& \partial_u \lambda \, = \, \partial_v \nu \, = \, \underbar{$\kappa$} \lambda \partial_r \left(1-\mu\right) \label{EvNuLambda} \, , \\ 
& \partial_u \varpi \, = \, \frac{\zeta^2}{2\underbar{$\kappa$}} \label{EvMass_u} \, , \\ 
& \partial_v \varpi \, = \, \frac{\theta^2}{2\lambda} \left(1-\mu\right) \label{EvMass_v} \, , \\ 
& \partial_u \theta \, = \, - \frac{\zeta \lambda}{r} \, , \label{wavetheta} \\
& \partial_v \zeta \, = \, - \frac{\theta \nu}{r} \, , \label{wavezeta} \\
& \partial_v \underbar{$\kappa$} \, = \, \frac{\underbar{$\kappa$} \theta^2}{r\lambda} \label{EvKbar} \, , \\
& \nu \, = \, \underbar{$\kappa$} \left(1-\mu\right) \, , \label{Kbar}
\end{align}
which is an overdetermined system of PDEs together with the algebraic equation~\eqref{Kbar} that replaces~\eqref{bark}. The results in~\cite{CGNS1}[Section~6] show that, under appropriate regularity conditions, this system is equivalent to the spherically symmetric EMS system, in a spacetime region where $\partial_vr=\lambda>0$.


\subsection{Reissner-Norstr\"om-de Sitter revisited}
\label{sectionRNdS}

If the scalar field vanishes identically, then the Hawking mass is constant, $\varpi\equiv M$. Among these electro-vacuum solutions~\footnote{Note that in this setting Birkhoff's Theorem is less restrictive and allows solutions which do not belong to the RNdS family, see~\cite{gajic}[Theorem 1.1] and~\cite{pedroPhd}.}, the Reissner-Nordstr\"om-de Sitter (RNdS) black holes, which form a 3-parameter $(M,e,\Lambda)$ family of solutions, are of special relevance to us here.  We will only consider the solutions of this family for which $M\geq0$ and the
function $r\mapsto (1-\mu)(M,r)$ has a positive root $r_c$
such that
\begin{equation}
\label{gravityCosm1}
\partial_r(1-\mu)(M,r_c)<0 \;,
\end{equation}
and
\begin{equation}
\label{gravityCosm2}
\partial_r^2(1-\mu)(M,r_c)<0\;.
\end{equation}
These conditions are enough to identify $r=r_c$ as a non-extremal~\footnote{Non-degenerate in the terminology of~\cite{lrr}.} cosmological horizon; in the case of a vanishing charge, the condition on the second derivative is in fact a consequence of $M\geq0$.  Note that, in RNdS, $r=r_c$ is a Killing horizon~\cite{lrr} whose generator is the spacetime's static Killing vector field.

The form of the metric of these RNdS solutions in double-null coordinates is not particularly elucidative and will be omitted. On the other hand, their global causal structure is crucial to us and is completely described by the Penrose diagram in Figure~\ref{figPenrose}; note that when we refer to a member of the RNdS family we are considering the solutions corresponding to the maximal (globally hyperbolic) development of appropriate Cauchy data. In this paper we will focus on the~\emph{cosmological region} $r>r_c$, which corresponds to the causal future of the cosmological horizon $r=r_c$; more precisely, we will be interested in understanding how stable its geometry is under non-linear spherically symmetric scalar perturbations.


\subsection{The characteristic initial value problem \IVP}
\label{sectionIVP}

The goal of this section is to establish an initial value problem for the system~\eqref{EvNuLambda}--\eqref{Kbar} which captures the essential features of a cosmological region to the future of a dynamic (i.e.\ not necessarily stationary) cosmological horizon $ \mathcal{C}^+ = \left\{ (u,v): v=0 \right\} $. We will refer to this as \IVP.

We fix the coordinate $u$ (up to a translation by a constant) by setting
\begin{equation} \label{minus1}
\underbar{$\kappa$}_0(u) := \uk(u,0) \equiv -1
\end{equation}
and, for a fixed $V>0$, we consider as initial data, given on
\begin{equation}
 \left[0,+\infty\right[ \times \left\{0\right\} \cup  \left\{0\right\} \times [0,V]\;,
\end{equation}
the functions
\begin{equation}
\begin{aligned}
r_0(v) &:= r(0,v) \, ,\\
\nu_0(u) &:= \nu(u,0) \, ,\\
\zeta_0(u) &:= \zeta(u,0) \, ,\\
\lambda_0(v) &:= \lambda(0,v) \, ,\\
\varpi_0(u) &:= \varpi(u,0) \, ,\\
\theta_0(v) &:= \theta(0,v) \, ,\\
\end{aligned}
\end{equation}
with $\nu_0$, $\lambda_0$, $\theta_0$ and $\zeta_0$ continuous, and $r_0$ and $\varpi_0$  continuously differentiable.

We also impose the sign conditions
\begin{equation} \label{signconditions}
\begin{aligned}
r_0(v) & > 0 \, ,\\
\tilde{r}_0(u) & > 0 \, ,\\
\lambda_0(v) & > 0 \, ,\\
\nu_0(u) & \geq 0 \, ,\\
\end{aligned}
\end{equation}
where
\begin{equation}
\tilde{r}_0(u):=r_0(0)+\int_0^u\nu_0(u')\,du' \;.
\end{equation}
The sign for $\nu_0$ is motivated by Hawking's area law (Proposition~\ref{area_Hawking}) and the sign of $\lambda_0$ by the need to model a cosmological (expanding) region; the sign of $\underbar{$\kappa$}_0$ was chosen to accommodate~\eqref{bark}.

The overdetermined character of our system implies that the initial data is necessarily constrained. To deal with this we assume the following compatibility conditions:
\begin{align}
\label{constraints}
r_0' &= \lambda_0\,, \\
\varpi_0' &= -\frac{\zeta_0^2}{2}\,,\\
\nu_0 &= - \left( 1 - \mu)(\varpi_0,\tilde{r}_0\right)\,.
\end{align}

We will also assume that our initial data approaches (without requiring any specific decay rate) a RNdS cosmological horizon along  $v=0$.
More precisely we will assume the existence of constants $r_c>0$ and $\varpi_c\geq0$ such that
\begin{equation}
\begin{aligned}
\lim_{u \to +\infty} {\tilde{r}_0(u)} & = r_c < +\infty \, , \\
\lim_{u \to +\infty} {\varpi_0(u)} & = \varpi_c < +\infty \, , \label{RNdS1}
\end{aligned}
\end{equation}
with
\begin{equation}
\label{rc_varpic}
(1-\mu)(r_c,\varpi_c)=0\;.
\end{equation}
The condition above captures the idea of approaching a RNdS horizon. To make sure
that it is a cosmological horizon we impose~\eqref{gravityCosm1} and~\eqref{gravityCosm2}, with $M=\varpi_c$. By continuity, given $\epsilon_0>0$, we can then shift our $u$ coordinate in such a way that
\begin{equation}
\label{smallMass}
|\tilde{r}_0(u)-r_c|+ |\varpi_0(u)-\varpi_c|<\epsilon_0\;\;,\text{ for all } u\geq0\;.
\end{equation}
Then,
we can choose $\epsilon_0$ small enough such that,
for all $u\geq 0$, we have
\begin{equation}
\label{RNdS2}
\partial_r (1-\mu) (\varpi_0(u),\tilde{r}_0(u))\leq -2k< 0
\end{equation}
for some $k>0$, which can be chosen arbitrarily close to $\frac{1}{2}\partial_r(1-\mu)(\varpi_c,r_c)$ by decreasing $\epsilon_0$.
%
%
This condition, together with \eqref{EvNuLambda}, \eqref{minus1} and the fact that the range of $u$ is unbounded imply that the line $v=0$ is in fact a complete null geodesic (compare with~\cite{CGNS2}[Section 8]).

Finally, we will assume that there exists $C_d>0$ such that
\begin{equation}
 \label{boundInitialZeta}
 |\zeta_0(u)|\leq C_d\;\;,\;\text{ for all }\; u\geq0\;.
\end{equation}

This finishes our characterization of \IVP.

\begin{Remark}
Note that the conditions~\eqref{RNdS1} together with the constraints~\eqref{constraints} impose restrictions on the integrability of our data: for instance, $\varpi_0(u)\rightarrow \varpi_c$ requires that  $\zeta_0\in L^2([0,\infty))$.
\end{Remark}
\begin{Remark}
The existence of ``large'' classes of data satisfying all the previous conditions, which is far from clear a priori, follows by a simple adaptation of the techniques in~\cite{CGNS4}[Section 3]; see~\cite{pedroPhd} for details. In alternative, one can also take the perspective of Theorem~\ref{thmMain} where our characteristic initial data arises from the evolution of appropriate Cauchy initial data. The existence of such Cauchy data, close to RNdS data, is expected to follow from the non-linear stability results in~\cite{Hintz,hintzVasy}.
\end{Remark}

We end this section with a simple adaptation of~\cite{CGNS1}[Theorem 4.4]:
\begin{Thm}[Maximal development and its domain ${\cal Q}$ for \IVP]
\label{maximal}
The characteristic initial value problem \IVP, with initial data as described above, has a unique classical solution (in the sense that all the partial derivatives occurring in~\eqref{EvNuLambda}--\eqref{Kbar} are continuous) defined on a maximal past set ${\cal Q}$ containing a neighborhood~\footnote{From now on, all topological statements will refer to the topology of $\left[0,+\infty\right[ \times\left[0,V\right]$ induced by the standard topology of $\mathbb{R}^2$.} of $ \left( \left[0,+\infty\right[ \times \{0\} \right) \cup \left( \{0\} \times \left[0,V\right] \right) $ in $\left[0,+\infty\right[ \times\left[0,V\right]$.
\end{Thm}


\section{Preliminary results and an extension criterion}


We start by establishing some basic sign properties, including a fundamental negative upper bound for $\partial_r(1-\mu)$~\eqref{redshiftbound} that corresponds to a global redshift.

\begin{lemma}\label{signsofquantities}
For the initial value problem \IVP, the following conditions are satisfied in ${\cal Q}$:
\begin{equation}
\label{lambdaSign}
 \lambda>0\;,
\end{equation}
\begin{equation}
\label{ukLower}
 \uk\leq-1\;,
\end{equation}
\begin{equation}
\label{nuSign}
 \nu \geq 0\,, \text{ with } \nu > 0 \text{ in } {\cal Q} \setminus \mathcal{C}^+\;,
\end{equation}
\begin{equation}
\label{lowerR}
 r \geq r(0,0)\;,
\end{equation}
\begin{equation}
 \partial_u \varpi \leq 0\;,
\end{equation}
\begin{equation}
 \partial_v \varpi \leq 0\;,
\end{equation}
\begin{equation}
\label{redshiftbound}
\partial_r(1-\mu) \leq -2k\;,\text{ for } k>0 \text{ as in~\eqref{RNdS2}}\,,
\end{equation}
\begin{equation}
\label{redshiftbound2}
 1-\mu \leq 0\,, \text{ with } 1-\mu < 0 \text{ in } {\cal Q} \setminus \mathcal{C}^+\;.
\end{equation}
\end{lemma}

\begin{proof}

Recall that $\lambda(0,v) > 0$ and notice how the evolution equation for $\lambda$ \eqref{EvNuLambda} implies that it cannot change sign (since the solution is an exponential times the initial data); that is, $\lambda>0$ in $\cal Q$. It then follows from~\eqref{EvKbar} and the fact that $\uk(u,0)\equiv-1$ that $\uk\leq -1$. As an immediate consequence~\eqref{EvMass_u} tells us that $\partial_u \varpi \leq 0$.

Now, let $\tilde{{\cal Q}} = \left\{ (u,v) \in {\cal Q}: \nu(u,\tilde{v}) \geq 0 \, , \forall \tilde{v} \in \left[ 0,v \right] \right\}$. In this subset, the signs of $\nu$ and $\lambda$ immediately imply that $r \geq r(0,0)$ and, since $ \nu = (1-\mu) \underbar{$\kappa$} $, one has $1-\mu \leq 0$ (with equality if and only if $\nu = 0$), and so $\partial_v \varpi \leq 0$ due to \eqref{EvMass_v}.

Furthermore,
\begin{equation*}
\begin{aligned}
& \partial_\varpi \partial_r (1-\mu) = \frac{2}{r^2} > 0 \, ,\\
\end{aligned}
\end{equation*}
and, for $v>0$,
$$\lambda > 0 \Rightarrow r(u,v) > r(u,0) \, .$$
By inspection of the graph of $r\mapsto 1-\mu(\varpi,r)$, we see that, for a sufficiently small $\epsilon_0$ in~\eqref{smallMass}, we have in $\tilde{{\cal Q}}$
\begin{equation}\label{redshiftbound3}
\begin{aligned}
\partial_r \left(1-\mu\right) \left(\varpi(u,v) , r(u,v)\right)  &  \leq \partial_r \left(1-\mu\right) \left(\varpi(u,0) , r(u,v)\right) \\
& \leq  \partial_r \left(1-\mu\right) \left(\varpi(u,0) , r(u,0)\right) \\
& \, \leq \,-2k \, ,
\end{aligned}
\end{equation}
where the last inequalities follow from~\eqref{RNdS2}.

Knowing the signs of the quantities above, \eqref{EvNuLambda} implies
\begin{equation}
\partial_v \nu = \underbar{$\kappa$} \lambda \partial_r \left(1-\mu\right) > 0 \label{EvNu}\;,
\end{equation}
in $\tilde{{\cal Q}}$. So, with the initial condition $\nu \geq 0$ on $\mathcal{C}^+$, one has $\nu > 0$ in $\tilde{{\cal Q}} \backslash \mathcal{C}^+$.

It now suffices to show that $\tilde{{\cal Q}} = {\cal Q}$. To this effect, define, for each $u \geq 0$, the sets $\mathcal{V}_u = \left\{ v: (u,v) \in {\cal Q} \right\}$ and $\tilde{\mathcal{V}}_u = \left\{ v: (u,v) \in \tilde{{\cal Q}} \right\}$. We will show that $\tilde{\mathcal{V}}_u$ is open and closed in $\mathcal{V}_u$.

Let $ \left\{ v_n \right\}_{n \in \mathbb{N}} \subset \tilde{\mathcal{V}}_u $ be a sequence with $v_n \to v^* \in \mathcal{V}_u$. If $\exists k \in \mathbb{N}$ such that $v_k > v^*$, then $\nu(u,v') \geq 0$ $\forall v' \leq v_k$ and, in particular, $\forall v' \leq v^*$; so, $(u,v^*) \in \tilde{{\cal Q}}$ and $v^* \in \tilde{\mathcal{V}}_u$. On the other hand, if there exists no such $k$, then there is an increasing subsequence $v_{n_m} \to v^*$; this means that $\left[0,v_{n_m}\right] \subset \tilde{\mathcal{V}}_u$, so $\bigcup_{m \in \mathbb{N}} \left[0,v_{n_m}\right] \subset \tilde{\mathcal{V}}_u$ and $\left[0,v^*\right[ \subset \tilde{\mathcal{V}}_u$, whence, by continuity, $\nu(u,v^*) \geq 0$ and, therefore, $v^* \in \tilde{\mathcal{V}}_u$. 

Suppose now that $v^* \in \tilde{\mathcal{V}}_u$, meaning $\nu(u,v') \geq 0$ $\forall v' \in [0,v^*]$; from \eqref{EvNu}, it follows that $\partial_v {\nu} (u,v^*) > 0$, and so, by continuity, $\exists \epsilon > 0$ such that $\partial_v {\nu} (u,v') > 0$ $\forall v' \in [v^*,v^*+\epsilon]$. This implies that $\nu(u,v') \geq 0$ $\forall v' \in [0,v^*+\epsilon]$ or, in other words, $(u,v^*+\epsilon) \in \tilde{{\cal Q}}$; thus, $\left[0,v^*+\epsilon\right[ \subset \tilde{\mathcal{V}}_u$. 

Since $\tilde{\mathcal{V}}_u\neq \varnothing$ (because $0 \in \tilde{\mathcal{V}}_u$) and $\mathcal{V}_u$ is connected (because ${\cal Q}$ is a past set), we have $\tilde{\mathcal{V}}_u=\mathcal{V}_u$ for all $u \geq 0$, whence $\tilde{{\cal Q}}={\cal Q}$.
\end{proof}

\begin{Remark} \label{confusao}
The Einstein equations do not propagate the sign of $\partial_vr=\lambda$ along a line of constant $v$, as can be attested by considering a line that goes from black hole exterior to the black hole interior of the RNdS solution. In contrast, the first order system~\eqref{EvNuLambda}-\eqref{Kbar} does propagate this sign; however, this is just a consistency check of the system, since the definition of \underbar{$\kappa$} requires that $\lambda$ does not change sign.    
\end{Remark}

Note that, in view of~\eqref{lambdaSign} and~\eqref{nuSign}, the curves of constant $r$ are spacelike in ${\cal Q}\setminus\{v=0\}$ and can be parameterized by
\begin{equation}
 u\mapsto (u,v_r(u))\;,
\end{equation}
or
\begin{equation}
\label{u_r}
 v\mapsto (u_r(v),v)\;,
\end{equation}
where $u_r$ and $v_r$ are $C^1$, but with domains unknown to us at the moment.

We can now easily improve some of the previous sign properties to more detailed bounds.
\begin{lemma}\label{bounds}
For the initial value problem \IVP, the following estimates hold in ${\cal Q}$:
\begin{equation}
-\frac{\Lambda}{6}r^3 + \frac{r}{2}+\frac{e^2}{2r} \, \leq \, \varpi \, \leq \, \varpi(0,0)\;, \label{boundsb}
\end{equation}
\begin{equation}
1-\frac{2\varpi(0,0)}{r} + \frac{e^2}{r^2} - \frac{\Lambda}{3} r^2 \, \leq \, 1-\mu \, \leq \, 0\;, \label{boundsc}
\end{equation}
\begin{equation}
-\Lambda r + \frac{1}{r} - \frac{e^2}{r^3} \, \leq \, \partial_r (1-\mu) \, \leq \, \frac{2\varpi(0,0)}{r^2} - \frac{2e^2}{r^3} - \frac{2\Lambda}{3} r\;. \label{boundse}
\end{equation}
%
\end{lemma}
\begin{proof}
%

Since we know that $\partial_u\varpi$, $\partial_v\varpi$ and $1-\mu$ are all non-positive, we get
\begin{equation}
1-\mu \leq 0 \Leftrightarrow -\frac{\Lambda}{6} r^3 + \frac{r}{2} + \frac{e^2}{2r} \leq \varpi
\end{equation}
and
\begin{equation}
\varpi \leq \varpi(0,0)  \Leftrightarrow
 1-\mu \geq 1 - \frac{2\varpi(0,0)}{r} + \frac{e^2}{r^2} - \frac{\Lambda}{3}r^2 \, .
\end{equation}
Applying the bounds for $\varpi$ to the identity
\begin{equation}
\partial_r (1-\mu) \, = \, \frac{2\varpi}{r^2} - \frac{2e^2}{r^3} - \frac{2\Lambda}{3}r \, ,
\end{equation}
leads to~\eqref{boundse}.
\end{proof}

Restricting ourselves  to a region with bounded radius function we can obtain further bounds:

\begin{lemma}\label{boundsR}
For the initial value problem \IVP and $R>\rz$ there exists a constant $C_{d,R}$, depending only on $R$ and on the size of the initial data, such that, in ${\cal Q}\cap\{r\leq R\}$,
\begin{equation}
\label{boundInR}
 \max\left\{\left|\frac{\theta}{\lambda}\right|, \left|\uk\right| , \left|\nu\right| , \left|\zeta\right|\right\} \leq C_{d,R}\;.
\end{equation}
\end{lemma}

\begin{proof}

Integrating~\eqref{EvMass_u} in ${\cal Q}\cap\{r\leq R\}$, while using the bounds~\eqref{boundsb} and~\eqref{lowerR}, we obtain
\begin{equation}
\int_0^u {\frac{\zeta^2}{- \underbar{$\kappa$} }} du =2\left(\varpi(0,v)-\varpi(u,v) \right)\leq  C_{d,R} \, .
\end{equation}
Then, integrating the evolution equation
\begin{equation}\label{thetaoverlambda}
\partial_u \left( \frac{\theta}{\lambda} \right) = - \frac{\zeta}{r} - \frac{\theta}{\lambda} \underbar{$\kappa$} \partial_r (1-\mu) \, ,
\end{equation}
yields, in view of the previous and~\eqref{redshiftbound},
\begin{equation*}
\begin{aligned}
\left| \frac{\theta}{\lambda} (u,v) \right| & \leq \left| \frac{\theta}{\lambda} (0,v) \right| e^{ - \int_0^u \uk \partial_r (1-\mu)(u',v) du'} + \int_0^u { \frac{|\zeta|}{r} (u'v)e^{ - \int_{u'}^u \underbar{$\kappa$} \partial_r (1-\mu) (u'',v)du''} du'} \label{tlbound} \\
&\leq  C_d + \frac{1}{R}\int_0^u { \frac{|\zeta|}{\sqrt{-\uk}} \sqrt{-\uk} \, e^{ - \int_{u'}^u \underbar{$\kappa$} \partial_r (1-\mu) du''} du'}  \\
&\leq C_d +\frac{1}{R} \left( \int_0^u { \frac{|\zeta|^2}{-\uk} (u',v) du' } \right)^\frac{1}{2} \left( \int_0^u { -\underbar{$\kappa$} \, e^{ 2 \int_{u'}^u -\underbar{$\kappa$} \partial_r (1-\mu) du''} du'} \right)^\frac{1}{2}  \\
&\leq C_d + C_{d,R}\left( \int_0^u { -\uk \, e^{ -4k \int_{u'}^u -\underbar{$\kappa$} du''} du'} \right)^\frac{1}{2}  \\
 &= C_d  + C_{d,R} \left( \left[ \frac{1}{4k} \, e^{ - 4k\int_{u'}^u |\underbar{$\kappa$}| du''} \right]_{u'=0}^{u'=u} \right)^\frac{1}{2} \\
&\leq C_{d,R} \;,
\end{aligned}
\end{equation*}
where $C_d$ represents a uniform bound on the initial data.

As a consequence, we are now  able to estimate $|\uk|$. In fact, using the evolution equation \eqref{EvKbar},

\begin{equation}\label{kbarr0}
\begin{aligned}
\left| \underbar{$\kappa$} (u,v) \right| \, & = \, \exp \left[ \int_0^{v} \left| \frac{\theta}{\lambda} \right|^2 \frac{\lambda}{r} (u,v') \, dv' \right]\\
& \leq \, \exp \left[ C_{d,R} \int_0^{v} \frac{\partial_v r}{r} (u,v') \, dv' \right] \\
& = \, \exp \left[ C_{d,R} \int_{r(u,0)}^{r(u,v)} \frac{dr'}{r'} \right] \\
& \leq \, \left[ \frac{r(u,v)}{r(0,0)} \right]^{C_{d,R}} \; .
\end{aligned}
\end{equation}

Combining this estimate with~\eqref{boundsc} immediately gives the bound for $|\nu|=\nu$.
Recalling that $\nu$ increases with $v$ in $\cal Q$, we can integrate~\eqref{wavezeta} and, using~\eqref{boundInitialZeta}, obtain the estimate
\begin{equation}\label{zetar0}
\begin{aligned}
 \left|\zeta(u,v)\right| \, & \leq \, \left|\zeta(u,0)\right| \, + \, \int_0^{v} \left( \left| \frac{\theta}{\lambda} \right| \frac{\lambda \nu}{r} \right) (u,v') \, dv' \\
& \leq  C_d \, + \, C_{d,R} \int_0^{v} \left( \frac{\lambda \nu}{r} \right) (u,v') \, dv'  \\
& \leq C_d\, + \, C_{d,R} \, \nu(u,v) \, \int_0^{v}  \frac{\lambda}{r}  (u,v') \, dv' \\
& \leq C_d\, + \, C_{d,R} \log \left| \frac{r(u,v)}{r(u,0)} \right|  \\
& \leq C_{d,R}\;.
\end{aligned}
\end{equation}

\end{proof}

As a result of all the estimates obtained in this section, we have uniform bounds in any region of the form ${\cal Q}\cap\{r\leq R\}$ for almost all quantities in terms of data and $R$. For example, from \eqref{EvMass_u} and \eqref{RNdS1} we easily obtain
\begin{equation}
\label{massAbove}
|\varpi (u,v)| \leq C_{d,R} \;.
\end{equation}
The single exception is $\lambda$, for which the best bound available in this region (compare with the upcoming~\eqref{lambdaLower}), is
\begin{equation}
\label{lambdaAbove}
 \lambda (u,v)  = \lambda (0,v) e^{ \int_0^u \underbar{$\kappa$} \partial_r \left( 1-\mu \right) du' } \leq C_d e^{C_{d,R}u} \;,
\end{equation}
which follows from~\eqref{EvNuLambda} and the bound for $\uk$. This is nonetheless enough to establish the following extension criterion, whose standard proof will be omitted.

\begin{Prop}
\label{extension}
For the initial value problem \IVP, let $p\in \overline{{\cal Q}}\setminus\{v=V\}$ be such that
$${\cal D}:=\left(J^-(p)\cap J^+(q)\right)\setminus\{p\}\subset {\cal Q}$$
for some point $q\in J^-(p)$. If the radius function is bounded in ${\cal D}$,
 \begin{equation}
 \label{rBounded}
  \sup_{\cal D}r<\infty\;,
 \end{equation}
then $p\in {\cal Q}$. In particular (since ${\cal Q}$ is an open past set), there exists $\epsilon>0$ such that
\begin{equation}
[0,u(p)+\epsilon]\times [0,v(p)+\epsilon]\subset {\cal Q}\;.
\end{equation}
\end{Prop}


\section{Unboundedness of the radius function}


Quite remarkably, we are almost ready to establish that, for the initial value problem \IVP, the radius function has to blow up along any null ray of the form $v=v_1\in(0,V]$. To do so we need the following lower bounds.

\begin{lemma}\label{boundsNuLambda}
For the initial value problem \IVP, we have in ${\cal Q}$
\begin{equation}
\label{lambdaLower}
\lambda (u,v)\geq \, \underbar{$\lambda_0$} \, e^{ 2k u } \;,
\end{equation}
where
\begin{equation}
\underbar{$\lambda_0$} \, := \, \inf_{v\in\left[0,V\right]} {\lambda (0,v)} \, > \, 0 \, .
\end{equation}
Moreover, there exists $\alpha>0$ such that, for any $v\in\left[0,V\right]$,
\begin{equation}
\label{boundsf}
\nu(u,v)\geq\alpha v\;.
\end{equation}
\end{lemma}

\begin{proof}
Note that in view of~\eqref{ukLower} and~\eqref{redshiftbound} we have in $\cal Q$
\begin{equation}
 \uk\, \partial_r (1-\mu)\geq 2k\;.
\end{equation}
It then follows that
\begin{equation}\label{lambda0}
\begin{aligned}
%
 \lambda (u,v) \, & = \, \lambda (0,v) \, \exp \left[ \int_0^u {\underbar{$\kappa$} \, \partial_r (1-\mu) \, (u',v) \, du'} \right] \\
& \geq \, \underbar{$\lambda_0$} \, \exp \left[ \int_0^u {2k\, du'} \right] \\
& = \, \underbar{$\lambda_0$} \, e^{ 2k u } \, .
\end{aligned}
\end{equation}
Consequently,
\begin{equation}
\begin{aligned}
\nu (u,v) \, & = \, \nu (u,0) + \int_0^v {\lambda \, \underbar{$\kappa$} \, \partial_r (1-\mu) \, (u,v') dv'} \\
& \geq \nu(u,0) + \underbar{$\lambda_0$} \, 2k \, e^{ 2k u } \, v \, ,
\end{aligned}
\end{equation}
from which~\eqref{boundsf} follows.

\end{proof}

We are now ready to obtain, in one go, two important characterizations of the behavior of the radius function:  we will show that the radius function blows up along any null line $v=v_1\neq0$, and that all the (spacelike curves) $r=r_1>r_c$ accumulate (in the topology of the plane) at $i^+=(\infty,0)$. More precisely:

\begin{Thm} \label{asymptotic}
Consider the initial value problem \IVP and let $v_1\in\left]0,V\right]$. Then,
\begin{equation}
\label{rUnbound}
 \sup \, \{ r(u,v_1) : (u,v_1) \in {\cal Q} \} = +\infty\;.
\end{equation}
Moreover, for all $r_1>r_c$ there exists $0<V_{r_1}\leq V$ such that the function $u_{r_1}$ (recall~\eqref{u_r})
is defined for all $v\in\left]0,V_{r_1}\right[$ and is such that
\begin{equation}
\label{uInftyV0}
u_{r_1}(v)\rightarrow\infty \;,\text{ as } v\rightarrow 0\;.
\end{equation}
\end{Thm}

\begin{proof}
From~\eqref{boundsf} we have
\begin{equation}
\begin{aligned}
r(u,v_1) \, &= \, r(0,v_1) + \int_0^u {\nu (u',v_1) \, du'} \\
& \geq \, r(0,v_1) + \alpha v_1u\;,
\end{aligned}
\end{equation}
which clearly diverges as $u\rightarrow\infty$, provided that $v_1\neq0$. The first result~\eqref{rUnbound} is now a consequence of the extension criterion (Proposition~\ref{extension}), and then the remaining conclusions follow.
\end{proof}

We will now show that, in fact, $r$ blows up in finite $u$-``time''.
\begin{Thm}\label{blowup}
Let ${\cal Q}$ be the domain of the maximal solution of the characteristic initial value problem with initial data satisfying \IVP.
Then:
\begin{enumerate}
 \item  For every $v_1\in\left(0,V\right]$ there exists $u^*(v_1)>0$ such that
\begin{equation}
 \lim_{u\rightarrow u^*(v_1)}r(u,v_1)=+\infty\;.
\end{equation}
\item For every $r_1>\rz$ there exists $0<V_{r_1}\leq V$ such that the function $u_{r_1}$
is defined for all $v\in(0,V_{r_1}]$. Moreover, $V_{r_1}=V$ and $u_{r_1}(V)<u^*(V)$ or $u_{r_1}(V_{r_1})=0$, i.e., the curve $r=r_1$ leaves $\cal Q$ to the right through $\left(\{u=0\}\cup\{v=V\}\right)\cap{\cal Q}$.
\item For any $r_1>r(0,V)$, $V_{r_1}=V$ and
\begin{equation}
{\cal Q}\cap\{r\geq r_1\}\setminus\{v=0\}=\{(u,v)\,:\, 0< v\leq V \text{ and } u_{r_1}(v)\leq u< u^*(v)\}\;.
\end{equation}
\item
\begin{equation}
\lim_{(u,v)\rightarrow (u^*(v_1),v_1)}r(u,v)=\infty\;.
\end{equation}
\end{enumerate}
\end{Thm}

\begin{proof}

From~\eqref{redshiftbound} and~\eqref{boundse} we can guarantee the existence of $\bar{A}>0$ such that
\begin{equation}
-\partial_r (1-\mu)  \geq  2\bar{A} \, r \, .
\end{equation}

 Therefore, recalling~\eqref{ukLower} and that $\nu(u,0) \geq 0$, the equation for $\nu$ yields
\begin{equation}
\begin{aligned}
\nu(u,v) \, & = \, \nu(u,0) \, + \, \int_0^v {\lambda \, \underbar{$\kappa$} \, \partial_r (1-\mu) \, (u,v') \, dv'} \\
& \geq \, 2\bar{A} \int_0^v {\lambda \, r \, (u,v') \, dv'} \\
& = \, \bar{A} \left[r^2 (u,v')\right]_{v'=0}^{v'=v} \\
& = \, \bar{A} \, r^2 (u,v) \, \left[ 1 - \frac{r^2 (u,0)}{r^2 (u,v)} \right] \, .
\end{aligned}
\end{equation}

So, for $r_c \, < \, r_1 \, \leq \, r(u,v)$, we have
\begin{equation}\label{lowerboundnu}
\nu(u,v)  \geq \, A \, r^2 (u,v) \, ,
\end{equation}
with
\begin{equation}
 A:=\bar{A} \left[ 1 - \frac{r_c^2}{{r_1}^2} \right] \;.
\end{equation}

In view of Theorem~\ref{asymptotic}, if $r_1$ is sufficiently large then for each $v\in\left]0,V\right]$ there exists $u_{r_1}(v)\geq0$ such that $r(u_{r_1}(v),v)=r_1$. Consequently,
\begin{equation}
\begin{aligned}
 &\int_{u_{r_1}(v)}^u \frac{\partial_{u}r(u',v)}{r^2(u',v)} du' \,  \geq \, A \, \left( u-u_{r_1}(v) \right)\\
&\Leftrightarrow \, \frac{1}{r_1} - \frac{1}{r(u,v)} \,  \geq \, A \, \left( u-u_{r_1}(v) \right) \\
&\Leftrightarrow \, r(u,v) \,  \geq \, \frac{1}{\frac{1}{r_1} - A \, \left( u-u_{r_1}(v) \right)} \, ,
\end{aligned}
\end{equation}
which tends to $+\infty$ when $u \to \left( u_{r_1}(v) + \frac{1}{Ar_1} \right)^-$. This establishes point 1.

Let $V_{r_1}$ be the supremum of the set of points in $\left]0,V\right]$ for which $u_{r_1}$ is well-defined in $\left]0,V_{r_1}\right[$.  If point 2 were not true then, in view of the spacelike character of $r=r_1$, we would have
$$\lim_{v\rightarrow V_{r_1}}u_{r_1}(v)=u^*(V_{r_1})\;.$$
Since $r$ is unbounded along $v=V_{r_1}$, there must exist $r_2>r_1$ and $0<u_2<u^*(V_{r_1})$ such that $r(u_2,V_{r_1})=r_2$. But since $(u_2,V_{r_1})\in\cal Q$, there exists $\epsilon>0$ such that $r(u_{r_2}(V_{r_1}-\epsilon), V_{r_1}-\epsilon)=r_2$ and $u_{r_2}(V_{r_1}-\epsilon)<  u_{r_1}(V_{r_1}-\epsilon)$, which contradicts the fact that $\nu>0$ in ${\cal Q}\setminus\{v=0\}$. This establishes point 2. Points 3 and 4 are then an immediate consequence.
\end{proof}


\section{Behavior of the scalar field for large $r$}
\label{sectionScalarAss}


It is possible to control the geometry of the curves of constant $r$ even further. To do this, we need precise knowledge about the behavior of the scalar field for large $r$, which is, of course, interesting in itself.

Note that, for $ r_1>r_c$, the function
\begin{equation}
 (u,v)\mapsto (r(u,v),v) \;,
\end{equation}
defines a diffeomorphism, mapping ${\cal Q}\cap\{r\geq r_1\}$ into $[r_1,\infty)\times(0,V]$, with inverse
\begin{equation}
 (r,v)\mapsto (u_r(v),v)\;.
\end{equation}
For any function $h=h(u,v)$ we define
\begin{equation}
 \tilde h(r,v): = h(u_r(v),v)\;.
\end{equation}
Now define the function $f:[r_1,\infty)\rightarrow\mathbb{R}$ by
\begin{equation}
\label{deff}
f(r):= \sup_{v\in(0,V]}\max\left\{ \left| {\frac{\tilde\theta}{\tilde\lambda}}(r,v) \right| , \left| {\frac{\tilde\zeta}{\tilde\nu}}(r,v) \right| \right\}\;;
\end{equation}
that this is a well-defined function follows from~\eqref{boundInR},~\eqref{lowerboundnu} and Theorem~\ref{blowup}.

\begin{Thm}
\label{decayScalar}
There exists $r_1>r_c$ such that, for all $r_2,r$ with $r_1\leq r_2\leq r$ we have, for the function defined by~\eqref{deff},
\begin{equation}
f(r)\leq f(r_2)\frac{r_2}{r}\;.
\end{equation}
\end{Thm}
\begin{proof}
We will start by establishing the following
\begin{lemma} \label{bootF}
There exists $r_1>r_c$ such that, given $r_2>r_1$, we have for any $r \geq r_2$
\begin{equation}
f(r)\leq C_2 \Rightarrow f(r) \leq \frac{C_2}{2}\left (1+\frac{{r_2}^2}{r^2}\right)\;.
\end{equation}
\end{lemma}
\begin{proof}
Integrating~\eqref{thetaoverlambda} from $r=r_2$ towards the future gives
\begin{equation}
\begin{aligned}
\frac{\theta}{\lambda} (u,v) =  &\frac{\theta}{\lambda}(u_{r_2}(v),v) e^{ - \int_{u_{r_2}(v)}^u \frac{\partial_{r}(1-\mu)}{1-\mu}\nu (u',v) du'}
\\
& - \int_{u_{r_2}(v)}^u { \frac{\zeta}{r}(u',v) e^{ - \int_{u'}^u \frac{\partial_{r}(1-\mu)}{1-\mu}\nu (u'',v) \, du''} \, du'} \;.
\end{aligned}
\end{equation}
To estimate the exponential we note that from Lemma~\ref{bounds} we have

\begin{equation}\label{frac1minusmu}
\begin{aligned}
\frac{\partial_r (1-\mu)}{1-\mu} \, &= \, \frac{-\partial_r (1-\mu)}{-(1-\mu)} \\
& \geq \, \frac{ \frac{2\Lambda}{3} r - \frac{2\varpi(0,0)}{r^2} + \frac{2e^2}{r^3} }{ \frac{\Lambda}{3} r^2 - 1 + \frac{2\varpi(0,0)}{r} - \frac{e^2}{r^2} } \\
& = \, \frac{2}{r} \, \frac{ 1 - \frac{3\varpi(0,0)}{\Lambda r^3} + \frac{3e^2}{\Lambda r^4} }{ 1 - \frac{3}{\Lambda r^2} + \frac{6\varpi(0,0)}{\Lambda r^3} - \frac{3e^2}{\Lambda r^4} } \\
& = \, \frac{2}{r} \, \left( 1 + \frac{ \frac{3}{\Lambda r^2} - \frac{9\varpi(0,0)}{\Lambda r^3} + \frac{6e^2}{\Lambda r^4} }{ 1 - \frac{3}{\Lambda r^2} + \frac{6\varpi(0,0)}{\Lambda r^3} - \frac{3e^2}{\Lambda r^4} } \right)  \\
& \geq \, \frac{2}{r} \, ,
\end{aligned}
\end{equation}
for any $r\geq r_1$ with $r_1$ sufficiently large.
As an immediate consequence, for $r_1\leq r_a\leq r_b$, we have
\begin{equation}\label{expintfrac1minusmu}
\exp \left[ - \int_{r_a}^{r_b} \frac{\partial_r (1-\mu)}{1-\mu} dr \right]  \leq \left(\frac{r_a}{r_b}\right)^2 \; .
\end{equation}
Using this estimate we find that
\begin{equation}
\begin{aligned}
\left| \frac{\theta}{\lambda}\right| (u,v)   \leq  & \left| \frac{\theta}{\lambda}\right| (u_{r_2}(v),v)  \left( \frac{r_2}{r} \right)^2  +\\
& + \frac{1}{r^2} \, \int_{u_{r_2}(v)}^u { \left( \left| \frac{\zeta}{\nu} \right|  \nu  r \right) (u',v)  du'} \; . \\
\end{aligned}
\end{equation}
Using the hypothesis of the lemma we get
\begin{equation}
\begin{aligned}
\left| {\frac{\tilde\theta}{\tilde\lambda}} (r,v) \right| & \leq  C_2 \left( \frac{r_2}{r} \right)^2 +  \frac{C_2}{r^2}  \int_{r_2}^r { r'  dr'} \\
& = C_2 \left( \frac{r_2}{r} \right)^2 + \frac{C_2}{r^2} \frac{r^2 - {r_2}^2}{2}\\
& =\frac{C_2}{2} \left( 1+ \frac{{r_2}^2}{r^2} \right) \;.
\end{aligned}
\end{equation}

Analogously, by integrating the equation
\begin{equation}
\partial_v \frac{\zeta}{\nu}= - \frac{\theta}{r} - \frac{\zeta}{\nu} \frac{\partial_v \nu}{\nu}\;,
\end{equation}
we obtain
\begin{equation}
\left| {\frac{\tilde\zeta}{\tilde\nu}} (r,v) \right| \, \leq \, \frac{C_2}{2} \left( 1+ \frac{{r_2}^2}{r^2} \right) \,,
\end{equation}
and the lemma follows.
\end{proof}
We proceed by showing that $f$ is non-increasing in $r\geq r_1$. For that purpose, assume there exist $r_1<r_2<r_3$ with $f(r_2)<f(r_3)$.
Then, there exists $r_*\in(r_2,r_3]$ such that
$$C_*:=\max_{r\in[r_2,r_3]}f(r)=f(r_*)\;.$$
Using Lemma~\ref{bootF} with $C_2=C_*$ we obtain the contradiction
$$C_*=f(r_*)\leq \frac{C_*}{2}\left (1+\frac{{r_2}^2}{{r_*}^2}\right)<C_* \, ,$$
and the claimed monotonicity follows.

As a consequence, we can in fact write, for all $r$ and $r_2$ such that $r_1\leq r_2\leq r$,
\begin{equation}\label{bestestimate}
f(r)\leq \frac{f(r_2)}{2} \left( 1+ \frac{{r_2}^2}{r^2} \right) \, .
\end{equation}
This can in turn be rearranged to yield
\begin{equation}
f(r)-f(r_2)  \leq\frac{f(r_2)}{2} \left( \frac{{r_2}^2}{r^2} - 1 \right) \, .
\end{equation}
Dividing by $r-r_2$ gives
\begin{equation}
\frac{f(r) - f(r_2)}{r -r_2} \leq \frac{f(r_2)}{2} \frac{{r_2}^2 - r^2}{r^2 ( r - r_2 )} \; .
\end{equation}
Since $f$ is a Lipschitz-continuous function, for almost every $r_2\geq r_1$,
\begin{equation}
\begin{aligned}
f'(r_2) \, & \leq \, \frac{f(r_2)}{2} \lim_{r \to r_2} \frac{{r_2}^2 - r^2}{r^2 ( r - r_2 )} \\
& =  -\frac{f(r_2)}{2} \lim_{r \to r_2} \frac{r + r_2}{r^2} \\
& =  -\frac{f\left(r_2\right)}{r_2} \;.
\end{aligned}
\end{equation}
Finally, integrating the last inequality we obtain, for all $r_1\leq r_2 \leq r$,
\begin{equation}
f(r)\leq f(r_2) \exp \left( -\int_{r_2}^r \frac{dr'}{r'} \right)= f(r_2) \frac{r_2}{r} \;.
\end{equation}
\end{proof}

As an immediate consequence we have
\begin{Cor}
\label{corGrad}
There exists $r_1>r_c$ and $C_1\geq 0$, depending on $r_1$, such that, in ${\cal Q}\cap\{r\geq r_1\}$,
\begin{equation}
\label{estGrad}
\max\left\{\left| {\frac{\theta}{\lambda}} \right| , \left| {\frac{\zeta}{\nu}} \right|\right\} \leq  \frac{C_1}{r}\;.
\end{equation}
\end{Cor}

\section{The causal character of $r=\infty$}

The estimates of the previous section will now allow us to obtain estimates for all geometric quantities in the region $r\geq r_1$, for $r_1>r_c$ sufficiently large. Those in turn will give the necessary information to show that $r=\infty$ is, in an appropriate sense, a spacelike surface.

We start by fixing an $r_1$ as in Corollary~\ref{corGrad}. Then, by integrating~\eqref{EvKbar} from $r=r_1$ while using~\eqref{kbarr0} and~\eqref{estGrad} we get
\begin{equation}
\begin{aligned}
 \left| \underbar{$\kappa$} (u,v) \right| \, & = \, \left| \underbar{$\kappa$} (u,v_{r_1}(u)) \right| \, \exp \left[ \int_{v_{r_1}(u)}^v \left| \frac{\theta}{\lambda} \right|^2 \frac{\lambda}{r} (u,v') \, dv' \right] \\
& \leq C_{d,r_1} \exp \left[ {C_1}^2 \int_{r_1}^r (r')^{-3} \, dr' \right] \\
& \leq  C_{d,r_1} \, .
\end{aligned}
\end{equation}
From Lemma~\ref{bounds} there exists $C>0$ such that
\begin{equation} \label{1-mu_bound}
\left| 1-\mu \right|  \leq  C \, r^2 \,,
\end{equation}
which, together with the previous estimate, immediately gives
\begin{equation}
\label{nuAbove}
\nu \leq C_{d,r_1} r^2\;.
\end{equation}

We also have, in view of~\eqref{lambdaLower} and~\eqref{expintfrac1minusmu},
\begin{equation}
\label{lambdaA1}
\begin{aligned}
\lambda(u,v) & =  \lambda(u_{r_1}(v),v) \, \exp \left[ \int_{u_{r_1}(v)}^{u} \frac{\nu \, \partial_r (1-\mu)}{1-\mu} (u',v) du' \right]
 \\
& \geq \underbar{$\lambda_0$} \, \left( \frac{r}{r_1} \right)^2 \;.\\
\end{aligned}
\end{equation}

From~\eqref{EvMass_v} we have
\begin{equation}
\varpi (u,v) \, = \, \varpi (u,v_{r_1}(u)) \, + \, \frac{1}{2} \int_{v_{r_1}(u)}^v \left( \frac{\theta}{\lambda} \right)^2 \, \lambda \, \left( 1-\mu \right) \, (u,v') \, dv' \, ,
\end{equation}
which we can estimate using \eqref{massAbove}, \eqref{estGrad} and \eqref{1-mu_bound}:
\begin{equation}
\label{massInft}
\begin{aligned}
\left| \varpi (u,v) \right|  & \leq \left| \varpi (u,v_{r_1}(u)) \right|  + \frac{1}{2} \int_{v_{r_1}(u)}^v \left| \frac{\theta}{\lambda} \right|^2  \lambda \left| 1-\mu \right|  (u,v') dv'
\\
& \leq C_{d,r_1} + C_{d,r_1} \int_{v_{r_1}(u)}^v \frac{1}{r^2} \lambda  r^2  (u,v')  dv'
\\
& \leq C_{d,r_1}  r\;.
\end{aligned}
\end{equation}
In view of this estimate, we have
\begin{equation}
\frac{\partial_r (1-\mu)}{1-\mu} =  \frac{2}{r} \, \left( 1 + \frac{ \frac{3}{\Lambda r^2} - \frac{9\varpi}{\Lambda r^3} + \frac{6e^2}{\Lambda r^4} }{ 1 - \frac{3}{\Lambda r^2} + \frac{6\varpi}{\Lambda r^3} - \frac{3e^2}{\Lambda r^4} } \right) \leq \frac{2}{r} + \frac{C_{d,r_1}}{r^2}\, ,
\end{equation}
for $r \geq r_1$ with $r_1$ sufficiently large. Using this inequality together with~\eqref{lambdaAbove} yields
\begin{equation}
\label{lambdaA2}
\begin{aligned}
\lambda (u,v)  & = \lambda (u_{r_1}(v),v) \, \exp \left[ \int_{u_{r_1}(v)}^{u} \frac{\nu \, \partial_r (1-\mu)}{1-\mu} (u',v) du' \right]  \\
& = \lambda(u_{r_1}(v),v) \, \exp \left[ \int_{r_1}^{r} \frac{\partial_r (1-\mu)}{1-\mu} (r',v) dr' \right] \\
& \leq C_{d,r_1}\exp\left({C_{d,r_1}u_{r_1}(v)}\right)\, \exp\left[ \log \left(\frac{r^2}{{r_1}^2}\right) + \frac{C_{d,r_1}}{r_1} \right] \\
& \leq C_{d,r_1}\exp\left({C_{d,r_1}u_{r_1}(v)}\right)\, r^2 \\
& \leq C_{d,r_1}\exp\left({C_{d,r_1}u}\right)\, r^2 \, .
\end{aligned}
\end{equation}

Next, we consider the level sets of $r$. Differentiating the identity
$$r = r(u_r(v),v)\;,$$
leads to
\begin{equation}
\begin{aligned}
 0  &=  \nu(u_r(v),v) \, u'_r(v)  +  \lambda(u_r(v),v) \\
&\Rightarrow  u'_r(v)  =  -\frac{\lambda(u_r(v),v)}{\nu(u_r(v),v)} \; .
\end{aligned}
\end{equation}

This means that, in view of~\eqref{lowerboundnu},~\eqref{nuAbove},~\eqref{lambdaA1} and~\eqref{lambdaA2}, we have the following result:
\begin{proposition}
\label{spacelike}
The curve $r=\infty$ is spacelike, in the sense that there exists $r_1>r_c$ and a constant $C_{d,r_1}>0$, depending only on $r_1$ and on the size of the initial data for \IVP, such that, for any $r\geq r_1$,
\begin{equation}
-C_{d,r_1}\exp\left({C_{d,r_1}u_{r_1}(v)}\right) < u'_r(v) < -C^{-1}_{d,r_1}\;,
\end{equation}
for all $v\in(0,V]$.
\end{proposition}

\begin{Remark}
Recall from~\eqref{uInftyV0} that $u_{r_1}(v)$ blows up as $v\rightarrow 0$.
\end{Remark}

For future reference, we collect some more estimates that can be easily derived from the information gathered until now.
From~\eqref{estGrad} and~\eqref{nuAbove} we have
\begin{equation}
| \zeta | =\left|\frac{\zeta}{\nu}\right||\nu|\leq C_{d,r_1}r\;.
\end{equation}
Analogously, \eqref{estGrad} and~\eqref{lambdaA2} give
\begin{equation}
\left| \theta \right| \leq C_{d,r_1}\exp\left({C_{d,r_1}u}\right) r \; .
\end{equation}
Note that, as an immediate consequence of the definitions of $\zeta$ and $\theta$, there exists $C>0$, such that in $\cal Q$
\begin{equation}
\label{boundDuPhi}
\left| \partial_u \phi \right|  \leq C\;,
\end{equation}
and
\begin{equation}
\label{boundDvPhi}
\left| \partial_v \phi \right|  \leq C e^{Cu}\; .
\end{equation}

\begin{Remark}
This last estimate is somewhat unpleasant and cruder than expected. Presumably one should be able to improve it by using the techniques of~\cite{CGNS4}, but this might require
imposing some extra decaying conditions on $\zeta_0(u)$. We will not pursue this here since~\eqref{boundDvPhi} turns out to be sufficient for our goals.
\end{Remark}


\section{Cosmic No-Hair}
\label{sectionNoHair}

Let $i=(u_{\infty},v_{\infty})$ be a point in $\{r=\infty\}\subset\overline{\cal Q}$.  Define a new coordinate $\tilde u$ by demanding that
\begin{equation}
 \tilde \uk(\tilde u,v_{\infty} )\equiv-1\;,
\end{equation}
where $\tilde\uk:=\frac{\tilde \nu}{1-\mu}$ and $\tilde \nu=\partial_{\tilde u}{r}$.
Then, $u=u(\tilde u)$  is such that
\begin{equation}
 \frac{du}{d\tilde u}=\frac{\tilde{\nu}(\tilde u, v_{\infty})}{\nu (u(\tilde u), v_{\infty})}=-\frac{(1-\mu)(\tilde u, v_{\infty})}{\nu (u(\tilde u)), v_{\infty})}\sim 1\;.
\end{equation}
The last estimate follows from the results of the previous section and in particular shows that the function $\tilde u=\tilde u(u)$ has a finite limit when $u\rightarrow u_\infty$.

In the region $r\geq r_1$, for $r_1$ as before, we can define $\kappa$ by
\begin{equation}
\label{kappa}
\kappa=\frac{\lambda}{1-\mu}\;.
\end{equation}
Analogously to the definition of $\tilde u$, we define $\tilde v=\tilde v(v)$ by demanding that
\begin{equation}
 \tilde \kappa(u_{\infty},\tilde  v )\equiv-1\;,
\end{equation}
which is equivalent to
\begin{equation}
 \frac{dv}{d\tilde v}=\frac{\tilde{\lambda}(u_{\infty},\tilde  v )}{\lambda (u_{\infty},v(\tilde  v) )}=-\frac{(1-\mu)(u_{\infty},\tilde  v )}{\lambda (u_{\infty},v(\tilde  v ))}\sim 1\;.
\end{equation}
We can now shift our coordinates so that $i=(\tilde u=0,\tilde v=0)$\;.

In these new coordinates the spacetime metric,
\begin{equation}
g=-\tilde \Omega^2(\tilde u,\tilde v) d\tilde ud \tilde v +r^2(\tilde u,\tilde v)\mathring{g}\;,
\end{equation}
satisfies
\begin{equation}
\label{omegaDym}
 \tilde \Omega^2(\hu,\hv)=-4\tilde \kappa \tilde \uk (1-\mu)(\hu,\hv)\;.
\end{equation}

Now consider a point $i_{dS}$ in the future null infinity of the pure de Sitter spacetime with cosmological constant $\Lambda$.
We can cover a neighborhood of the past of this point by null coordinates for which $i_{dS}=(\tilde u=0,\tilde v=0)$ and the metric takes the form
\begin{equation}
^{dS}\!g=- \Omega^2_{dS}(\tilde u,\tilde v)d\tilde ud \tilde v +r_{dS}^2(\tilde u,\tilde v)\mathring{g}\;,
\end{equation}
with
\begin{equation}
\label{omegaDS}
\tilde \Omega_{dS}^2(\tu,\tv)=- 4\left(1-\frac{\Lambda}{3}r^2_{dS}(\tilde u,\tilde v)\right)\;.
\end{equation}
We identify $J^-(i)\cap\{r\geq r_1\}\subset\{(\hu,\hv)\,:\, \hu\leq 0\;, \hv \leq 0)\}\setminus\{(0,0)\}$, in our dynamic solution $({\cal Q}, g)$, with the region of de Sitter space with the same range of the $(\hu,\hv)$ coordinates.

The first goal of this section is to show that  the components of (the dynamic spacetime metric) $g$ approach those of the de Sitter metric $^{dS}\!g$ in
$J^-(i)\cap\{r\geq r_2\}$ as $r_2\rightarrow\infty$. With that in mind, we start by establishing the following
\begin{lemma}
 In $J^-(i)\cap\{r\geq r_2\}$\;, for sufficiently large $r_2$,
\begin{equation}
\label{tildeks}
 -1\leq \tilde\uk\;, \tilde \kappa \leq -1+\frac{C}{{r_2}^2}\;.
\end{equation}
\end{lemma}

\begin{proof}
 Concerning $\tilde\uk$, the estimate from below is an immediate consequence of the fact that $\partial_{\tilde v}\tilde \uk\leq 0$ (see \eqref{EvKbar}).
 To establish the other inequality, recall Theorem~\ref{decayScalar}, and note that
 $$\left|\frac{\tilde \theta}{\tilde \lambda}\right|(\tu,\tv) \leq f(r(\tu,\tv))\leq f(r_2)\frac{r_2}{r(\tu,\tv)}\;,$$
 with
 $$f(r_2)\leq f(r_1)r_1\frac{1}{r_2}=:C_1\frac{1}{r_2}\;.$$
 Then, for $\tv\leq 0$ we get
\begin{equation}
\begin{aligned}
  \tilde \uk (\tu,\tv)  & = \tilde \uk (\tu,0) \exp \left(- \int_{\tv}^0 \left| \frac{\tilde \theta}{\tilde \lambda} \right|^2 \frac{\tilde \lambda}{r} (\tu,v') \, dv' \right)
  \\
& \leq -\exp \left(-f^2(r_2){r_2}^2 \int_{\tv}^0 \frac{\tilde \lambda}{r^3} (\tu,v') \, dv' \right)
\\
& \leq  -\exp \left(-f^2(r_2)\frac{{r_2}^2}{2} \left(\frac{1}{r^2(\tu,\tv)}- \frac{1}{r^2(\tu,0)} \right) \right)
\\
& \leq -\exp \left(-f^2(r_2) \right)
\\
& \leq-\exp \left(-\frac{C}{{r_2}^2} \right)
\\
&\leq -1+\frac{C}{{r_2}^2} \;.
\end{aligned}
\end{equation}

The results concerning $\tilde \kappa$ follow in a dual manner by using the equation
\begin{equation}
 \partial_{\tu}\tilde \kappa =\tilde \kappa\left( \frac{\tilde \zeta}{\tilde \nu}\right)^2\frac{\tilde \nu}{r}\;.
\end{equation}
\end{proof}

As an immediate consequence, in $J^-(i)\cap\{r\geq r_2\}$,
\begin{equation}
\label{1-muAB1}
-\left(1-\frac{C}{{r_2}^2}\right)(1-\mu)\leq\tilde \lambda, \tilde \nu  \leq -(1-\mu)\;.
\end{equation}
Since, in view of~\eqref{massInft}, we have 
$$-(1-\mu)\leq C+\frac{\Lambda}{3}r^2=\left(1+\frac{3C}{\Lambda r^2}\right)\frac{\Lambda}{3}r^2\leq \left(1+\frac{C}{ {r_2}^2}\right)\frac{\Lambda}{3}r^2$$
and, similarly,
$$-(1-\mu)\geq \left(1-\frac{C}{ {r_2}^2}\right)\frac{\Lambda}{3}r^2\; ,$$
estimates~\eqref{1-muAB1} give rise to
\begin{equation}
\label{gradR}
\left(1-\frac{C}{{r_2}^2}\right)\frac{\Lambda}{3}r^2(\tu,\tv)\leq \tilde \lambda, \tilde \nu  \leq \left(1+\frac{C}{{r_2}^2}\right)\frac{\Lambda}{3}r^2(\tu,\tv)\; ,
\end{equation}
for sufficiently large $r_2$ and $r>r_2$.

For the de Sitter spacetime we have
\begin{equation}
\label{gradRdS}
\left(1-\frac{C}{{r_2}^2}\right)\frac{\Lambda}{3}r_{dS}^2(\tu,\tv)\leq \tilde \lambda_{dS}, \tilde \nu_{dS}\leq \frac{\Lambda}{3} r_{dS}^2(\tu,\tv)\;.
\end{equation}

\begin{lemma}
 Let $A,B$ be positive constants and $r=r(u,v)$ a function in $J^-(0,0)=\{(u,v)\,:\, u\leq 0\;, v \leq 0)\}\setminus\{(0,0)\}$ satisfying
\begin{equation}
Br^2(u,v)\leq \partial_vr , \partial_ur  \leq A r^2(u,v)\;,
\end{equation}
and
\begin{equation}
r(x,0)\rightarrow \infty\;\;,\;\; r(0,x)\rightarrow \infty,\;\text{ as } x\rightarrow 0\;.
\end{equation}
Then, in $J^-(0,0)$,
\begin{equation}
-\frac{1}{A(u+v)}\leq r(u,v)\leq -\frac{1}{B(u+v)}\;.
\end{equation}
\end{lemma}

\begin{proof}
We have
\begin{equation}
\begin{aligned}
\int_u^0 \frac{1}{r^2(u',v)}\frac{\partial r}{\partial u}(u',v) du'\leq \int_u^0  A du'
&\Rightarrow
\int_{r(u,v)}^{r(0,v)} \frac{1}{r^2}dr \leq -Au
\\
&\Rightarrow
  \frac{1}{r(u,v)}-\frac{1}{r(0,v)}\leq -Au \;
\end{aligned}
\end{equation}
and
\begin{equation}
\begin{aligned}
\int_v^{v_1} \frac{1}{r^2(0,v')}\frac{\partial r}{\partial v}(0,v') dv'\leq \int_v^{v_1}  A dv'
&\Rightarrow
\int_{r(0,v)}^{r(0,v_1)} \frac{1}{r^2}dr \leq A(v_1-v)
\\
&\Rightarrow
  \frac{1}{r(0,v)}-\frac{1}{r(0,v_1)}\leq A(v_1-v)
  \\
&\Rightarrow
  \frac{1}{r(0,v)}\leq -Av\;,
\end{aligned}
\end{equation}
where the last implication follows by taking the limit $v_1\rightarrow 0$.
Then,
 $$\frac{1}{r(u,v)}\leq -Au+\frac{1}{r(0,v)}\leq -A(u+v)\;,$$
and the result follows. The other bound is similar.
\end{proof}

In view of~\eqref{gradR},~\eqref{gradRdS} and the previous lemma, we conclude that, in $J^-(i)\cap\{r\geq r_2\}$,
\begin{equation}
\label{rMinusR}
 |r(\tu,\tv)-r_{dS}(\tu,\tv)|\leq -\frac{C}{{r_2}^2}\frac{1}{\tu+\tv}\leq \frac{C}{{r_2}^2}\, r_{dS}(\tu,\tv)\;.
\end{equation}
Then it is clear that
\begin{equation}
\label{rMinusR2}
\begin{aligned}
\frac{|r^2(\tu,\tv)-r^2_{dS}(\tu,\tv)|}{r^2_{dS}(\tu,\tv)}
&\leq \frac{C}{{r_2}^2}\frac{1}{r_{dS}(\tu,\tv)} (r(\tu,\tv)+r_{dS}(\tu,\tv))
\\
& \leq \frac{C}{{r_2}^2}\left(\frac{r}{r_{dS}}(\tu,\tv)+1\right)
\\
&\leq \frac{C}{{r_2}^2}\;.
\end{aligned}
\end{equation}
Recalling~\eqref{omegaDym} and~\eqref{omegaDS}, we obtain, arguing as before,
\begin{equation}
\label{omegaAB}
\left(1-\frac{C}{{r_2}^2}\right)\frac{4\Lambda}{3}r^2(\tu,\tv)\leq \tilde \Omega^2(\tu,\tv) \leq \left(1+\frac{C}{{r_2}^2}\right)\frac{4\Lambda}{3}r^2(\tu,\tv)\;,
\end{equation}
from which we see that
\begin{equation}
\label{OMinusO}
 \frac{|\tilde \Omega^2(\tu,\tv)-\tilde \Omega^2_{dS}(\tu,\tv)|}{\tilde \Omega^2_{dS}(\tu,\tv)}\leq \frac{C}{{r_2}^2}\;.
\end{equation}

Let us now rewrite these last estimates in terms of an orthonormal frame $\{e_a,e_A\}_{a\in\{0,1\},A\in\{3,4\}}$ of de Sitter, where $\{e_A\}$ is an orthonormal frame of $r^2_{dS}\mathring{g}$  and
\begin{eqnarray}
e_0&=&\tilde\Omega_{dS}^{-1}(\partial_{\tu}+\partial_{\tv})\;, \\
e_1&=&\tilde\Omega_{dS}^{-1} (-\partial_{\tu}+\partial_{\tv})\;.
\end{eqnarray}
In the frame above the dynamic metric reads
\begin{equation}
 g_{ab}=\frac{\tilde\Omega^2}{\tilde\Omega^2_{dS}}\eta_{ab}
\end{equation}
and
\begin{equation}
 g_{AB}=\frac{r^2}{r^2_{dS}}\delta_{AB}\;.
\end{equation}
So we see that, in this frame, \eqref{rMinusR2} and~\eqref{OMinusO} imply that
\begin{equation}
\sup_{J^-(i)\cap\{r\geq r_2\}}|g_{\mu\nu}-^{dS}\!\!g_{\mu\nu}|\lesssim r_2^{-2}\;.
\end{equation}

Let us now consider the Riemann curvature, whose components satisfy~\cite{dafRen}[Appendix A]
\begin{eqnarray}
R^a_{bcd}&=&K\left(\delta^a_cg_{bd}-\delta^a_dg_{bc}\right)\;,
\\
R^a_{BcD}&=&-r^{-1}\nabla^a\nabla_crg_{BD}\;,
\\
R^A_{BCD}&=&r^{-2}\left(1-\partial_ar\partial^ar\right)\left(\delta^A_Cg_{BD}-\delta^A_Dg_{BC}\right)\;.
\end{eqnarray}
The Gaussian curvature of the surfaces of fixed angular coordinates is given by~\cite{dafRen}[Appendix A]
\begin{equation}
\label{gaussCurv}
 K=4\tilde\Omega^{-2}\partial_{\tu}\partial_{\tv}\log\tilde\Omega\;.
\end{equation}
Recalling~\eqref{wave_Omega} and using~\eqref{boundDuPhi},~\eqref{boundDvPhi} and~\eqref{omegaAB}, we first obtain
\begin{equation}
\label{gaussCurv2}
 K=4\tilde\Omega^{-2}\left(O(1)+\frac{\tilde \lambda\tilde \nu}{r^2}\right)\;,
\end{equation}
and then, using~\eqref{gradR} and~\eqref{omegaAB} again, we get, for $r>r_2$,
\begin{equation}
\label{KAB}
\left(1-\frac{C}{{r_2}^2}\right)\frac{\Lambda}{3}\leq K\leq \left(1+\frac{C}{{r_2}^2}\right)\frac{\Lambda}{3}\;.
\end{equation}
The previous expression is valid in both spacetimes, consequently,
$$\left|K-K_{dS}\right|\lesssim r_2^{-2}\;.$$
In the $\{e_{\mu}\}$ frame we then have
\begin{eqnarray*}
\left| R^a_{bcd}-^{dS}\!\!R^a_{bcd} \right| & = & \left|\frac{\tilde \Omega^2}{\tilde \Omega_{dS}^2} K-K_{dS}\right|\left|\delta^a_c\eta_{bd}-\delta^a_d\eta_{bc}\right|
\\
&\lesssim& \left|\frac{(\tilde\Omega^2-\tilde\Omega_{dS}^2)}{\tilde\Omega_{dS}^2}K\right|+\left|K-K_{dS}\right|
\\
&\lesssim& r_2^{-2}\;.
\end{eqnarray*}
%
%
%

The Hessian of the radius function is~\cite{dafRen}[Appendix A]\footnote{There is a difference in a factor of $4\pi$ due to our different convention for the value of Newton's gravitational constant.}
\begin{equation}
  \nabla_a\nabla_b r=\frac{1}{2r}\left(1-\partial_cr\partial^cr\right)g_{ab}- r(T_{ab}-\tr{T}g_{ab})-\frac{1}{2}r\Lambda g_{ab} \, ,
\end{equation}
where $\tr{T} = g^{ab}T_{ab}$. We have
$$\partial_cr\partial^cr=1-\mu\;,$$
and, in view of~\eqref{boundDuPhi} and~\eqref{boundDvPhi},
\begin{equation}
\label{ephi}
|e_a\phi|\lesssim r_{dS}^{-1}\;.
\end{equation}
Using this inequality, we can estimate the scalar field part of the energy momentum tensor by
$$|(e_a\phi) (e_b\phi)+2(\partial_{\tilde u}\phi)(\partial_{\tilde v}\phi) \Omega_{dS}^{-2} \eta_{ab}|\lesssim r_{dS}^{-2}\;,$$
and the electromagnetic part is, in view of~\eqref{energymomentumtensor} and~\eqref{F}, bounded by $Cr^{-4}$. Consequently,
$$|T_{ab}|\lesssim r^{-4}+r_{dS}^{-2}\;.$$
Since $\tr{T}=0$, we get
$$|\nabla_a\nabla_b r|\lesssim r\;,$$
and
\begin{eqnarray*}
\left| \nabla_a\nabla_b r- ^{dS}\!\nabla_a^{dS}\!\nabla_b r_{dS} \right|
 &\leq&
 \frac{\Lambda}{6}
 \left[|r-r_{dS}|+ \frac{C}{{r_2}^2}(r+r_{dS})\right]
 \\
 && + Cr \left|r^{-4}+r_{dS}^{-2}\right| + \frac{\Lambda}{2}\left|\frac{\tilde\Omega^2}{\tilde\Omega_{dS}^2}r-r_{dS}\right|
\\
&\lesssim& \frac{r+r_{dS}}{{r_2}^2}\;,
\end{eqnarray*}
where we used the fact that $r,r_{r_{dS}}\gtrsim r_2$, together with~\eqref{rMinusR} and~\eqref{OMinusO}.
%
%

We are now ready to estimate
\begin{eqnarray*}
\left| R^a_{BcD}-^{dS}\!\!R^a_{BcD} \right|
&=&
\left|\left(r^{-1}_{dS}\ ^{dS}\nabla_a^{dS}\nabla_c r_{dS}-r^{-1}\frac{\tilde\Omega^2_{dS}}{\tilde\Omega^2}\frac{r^2}{r_{dS}^2}\nabla_a\nabla_c r\right)\delta_{BD}\right|
\\
&\leq&
\frac{1}{r_{dS}}\left|^{dS}\nabla_a^{dS}\nabla_c r_{dS}-\nabla_a\nabla_c r + \left(1-\frac{\tilde\Omega^2_{dS}}{\tilde\Omega^2}\frac{r}{r_{dS}}\right)\nabla_a\nabla_c r\right|
\\
&\lesssim&
\frac{1}{r_{dS}}\left[\frac{r+r_{r_{dS}}}{{r_2}^2}+\frac{r}{{r_2}^2}\right]
\\
&\lesssim& r_2^{-2}\;.
\end{eqnarray*}
%
%
%

The previous estimates and a similar reasoning also allow us to conclude that 
\begin{equation}
 |R^A_{BCD}-^{dS}\!\!R^A_{BCD}|\lesssim r_2^{-2}.
\end{equation}
%
%


\section{Proofs of Theorems~\ref{thmMain} and~\ref{thmMainGlobal} }
\label{sectionProofs}

We can now easily prove Theorem~\ref{thmMain}:

\begin{proof}
We can use condition \emph{(iii)} and the existing gauge freedom $u\mapsto f(u)$ to impose the condition $\uk_0:=-\frac{1}{4}\Omega^2(\partial_vr)^{-1}|_{v=0}\equiv-1$. Then there exist  $U,V>0$ such that the conditions of \IVP of Section~\ref{sectionIVP} are satisfied on $\left[U,+\infty\right[ \times \left\{0\right\} \cup  \left\{0\right\} \times [0,V]$ (in Section~\ref{sectionIVP} we translate $u$ so that $U= 0$): the regularity conditions are a consequence of the smoothness of the Cauchy initial data; the sign conditions~\eqref{signconditions} follow, using continuity, from the hypotheses and Proposition~\ref{area_Hawking};  the constraints from the fact that $({\cal M},g,F,\phi)$ is a solution of the Einstein-Maxwell-scalar field system; and the subextremality conditions from the fact that the data is converging to the reference subextremal RNdS.

Then the  results \emph{1-4} are just a compilation of the results in Theorem~\ref{blowup}, Corollary~\ref{corGrad}, Proposition~\ref{spacelike} and Section~\ref{sectionNoHair}.
\end{proof}

We end with the proof of Theorem~\ref{thmMainGlobal}

\begin{proof}
Conditions \emph{i'-iv'} together with assumption {\em I} were tailored to ensure that \IVP holds in $\left[0,+\infty\right[ \times \left\{0\right\} \cup  \left\{0\right\} \times [0,V]$, for all $V\geq 0$, and that the dual initial value problem, obtained by interchanging $u$ and $v$, holds in
$\left[0,U\right] \times \left\{0\right\} \cup  \left\{0\right\} \times [0,\infty)$, for all $U>0$. Then, as in the proof of Theorem~\ref{thmMain}, the results are valid in
${\cal Q}={\cal Q}_1\cup {\cal Q}_2$, where ${\cal Q}_1=\tilde {\cal Q}\cap [0,\infty)\times[0,V)$ and ${\cal Q}_2=\tilde {\cal Q}\cap [0,U)\times[0,\infty)$. Choosing $U$ and $V$ sufficiently large we can ensure that ${\cal Q}=\tilde {\cal Q}\cap [0,\infty)^2$, as desired.

Conditions~\emph{i'-iv'} allow us to invoke conclusion {\em 3} of Theorem~\ref{thmMain} to conclude that, under assumption {\em II}, we have that ${\cal N}_1$ is in fact empty; by the interchange symmetry between $u$ and $v$, we also conclude that ${\cal N}_2=\varnothing$. Then, the future boundary of ${\cal Q}=\tilde {\cal Q}\cap [0,\infty)^2$  is only composed of  ${\cal B}_{\infty}$, so that conclusion {\em 1} holds throughout, and the remaining conclusions are valid near the points $(\infty,v_{\infty})$ and $(u_{\infty},\infty)$. Now, Raychaudhuri's equations imply that in
$D^{-}({\cal B}_{\infty})={\cal Q}$  we have $\lambda>0$, $\nu>0$ and, consequently $1-\mu<0$. From the previous signs we see that the conclusions of Lemma~\ref{bounds} hold and hence~\eqref{frac1minusmu} is valid in ${\cal Q}\cap \{r\geq r_1\}$, for a sufficiently large $r_1$. To see that the conclusions of Section~\ref{sectionScalarAss} then hold in $\cal Q$ we just need to note that the function $f$ defined in that section remains well-defined in the new (``larger'')  $\cal Q$; but this follows from compactness, after noting that the curve of constant $r$ where one takes the supremum is the union of a subset near $(\infty,v_{\infty})$, a subset near $(u_{\infty},\infty)$, and a compact subset connecting these two. This shows that {\em 2} holds in $\cal Q$. The remaining conclusions now follow, since they are rooted in {\em 2} and in the fact that all relevant quantities have appropriate bounds along $r=r_1$, which follows by a compactness argument as before.
 \end{proof}

\section*{Acknowledgements}

We thank Anne Franzen for sharing and allowing us to use Figure~\ref{figPenrose}.
This work was partially supported by FCT/Portugal through UID/MAT/04459/2013 and grant (GPSEinstein) PTDC/MAT-ANA/1275/2014.
Pedro Oliveira was supported by FCT/Portugal through the LisMath scholarship PD/BD/52640/2014.

%


\begin{thebibliography}{}

\bibitem{AH}
H.\ Andr\'{e}asson and H.\ Ringstr\"om, \emph{Proof of the cosmic no-hair conjecture in the T3-Gowdy symmetric Einstein-Vlasov setting}, JEMS (to appear). 

\bibitem{Beyer}
F.\ Beyer, \emph{The cosmic no-hair conjecture: a study of the Nariai solutions}, Proceedings of the
Twelfth Marcel Grossmann Meeting on General Relativity (R. Ruffini T. Damour, R.T. Jantzen,
ed.), 2009. 

\bibitem{brillRNdS}
D. R. Brill and S. A. Hayward, \emph{Global Structure of a Black-Hole Cosmos and its Extremes}, Class.\ Quantum Grav.\ {\bf 11} (1994) 359-370. 



\bibitem{lrr} P.~Chru\'sciel, J.~Costa and M.~Heusler, \emph{Stationary Black Holes: Uniqueness and Beyond},  Living Reviews in Relativity {\bf 12} (2012).

\bibitem{Costa:2013}
J.~L.~Costa,
\emph{The spherically symmetric Einstein-scalar field system with positive and vanishing cosmological constant: a comparison},
Gen.\ Rel.\ Grav.\ {\bf 45} (2013) 2415-2440. 

\bibitem{CAN}
J.~L.~Costa, A.\ Alho and J.\ Nat\'ario,
\emph{The problem of a self-gravitating scalar field with positive cosmological constant},
Ann.\ Henri Poincar\'e {\bf 14} (2013) 1077-1107. 

\bibitem{CGNS1}
{J.~L~.~Costa, P.~Gir\~ao, J.~Nat\'ario and J.~Silva},  \emph{On the global uniqueness for the Einstein-Maxwell-scalar field system with a cosmological constant.
Part 1.  Well posedness and breakdown criterion},  Class.\ Quantum Grav.\ {\bf 32} (2015) 015017.

\bibitem{CGNS2} {J.~Costa, P.~Gir\~ao, J.~Nat\'ario and J.~Silva},  \emph{On the global uniqueness for the Einstein-Maxwell-scalar field system with a cosmological constant. Part 2. Structure of the solutions and stability of the Cauchy horizon}, Commun.\ Math.\ Phys.\ {\bf 339} (2015) 903-947.

\bibitem{CGNS4}
{J.~Costa, P.~Gir\~ao, J.~Nat\'ario and J.~Silva},  \emph{On the occurrence of mass inflation for the Einstein-Maxwell-scalar field system with a cosmological constant and an exponential Price law},   arXiv:1707.08975.

\bibitem{JPZ}
{J.~L.~Costa, J.~Nat\'ario and P.~Oliveira},  \emph{Decay of solutions of the wave equation in expanding cosmological spacetimes},  (in preparation).

\bibitem{dafRen}
M.\ Dafermos and A.\ Rendall \emph{Strong cosmic censorship for surface-symmetric cosmological spacetimes with collisionless matter}, Comm.\ Pure Appl.\ Math.\ {\bf 69} (2016) 815–908. 

\bibitem{Friedrich:1986}
 H.\ Friedrich, \emph{On the existence of n-geodesically complete or future complete solutions of
Einstein’s field equations with smooth asymptotic structure}, Commun.\ Math.\ Phys.\ {\bf 107} (1986) 587–609.


\bibitem{Friedrich:2016}
H.\ Friedrich, \emph{Geometric Asymptotics and Beyond}, arXiv:1411.3854.

\bibitem{gajic}
D.\ Gajic, \emph{Linear waves on constant radius limits of cosmological black hole spacetimes} (2014), to appear in Adv.\ Theor.\ Math.\ Phys.\, arXiv:1412.5190v2.

\bibitem{hintzVasy} P.~Hintz and A.~Vasy \emph{The global non-linear stability of the Kerr-de Sitter family of black holes}, arXiv:1606.04014.

\bibitem{Hintz} P.~Hintz \emph{Non-linear stability of the Kerr-Newman-de Sitter family of charged black holes}, arXiv:1612.04489.

\bibitem{kroonBook}
J.\ V.\ Kroon, \emph{Conformal Methods in General Relativity}, Cambridge University Press (2016).

\bibitem{lubbeKroon}
C.\ Lubbe and J.\ V.\ Kroon, \emph{A conformal approach for the analysis of the nonlinear
stability of pure radiation cosmologies}, Ann.\ Phys.\ {\bf 328} (2013) 1-25. 

\bibitem{Oli:2016}
T.\ Oliynyk, \emph{Future stability of the FLRW fluid solutions in the presence of a positive cosmological constant} Commun.\ Math.\ Phys.\ {\bf 346} (2016) 293-312. 

\bibitem{pedroPhd} Pedro F. C. Oliveira, Phd thesis, IST-ULisboa (in preparation).

\bibitem{Rad}
K.\ Radermacher \emph{On the Cosmic No-Hair Conjecture in T2-symmetric non-linear scalar field spacetimes}, arXiv:1712.01801.

\bibitem{Rendall:2004}
A.\ Rendall, \emph{Asymptotics of solutions of the Einstein equations with positive cosmological
constant}, Ann.\ Henri Poincar\'e 5 (2004), 1041–1064, arXiv:gr-qc/0312020.

\bibitem{RingUniv}
H.\ Ringstr\"om, \emph{On the Topology and Future Stability of the Universe}, Oxford University Press (2013).

\bibitem{Rodnianski:2009}
I.\ Rodnianski and J.\ Speck, \emph{The stability of the irrotational Euler-Einstein system with
a positive cosmological constant}, J.\ Eur.\ Math.\ Soc.\ {\bf 15} (2013) 2369-2462. 

\bibitem{schlue}
V.\ Schlue \emph{Global results for linear waves on expanding Kerr and Schwarzschild de Sitter cosmologies},
Commun.\ Math.\ Phys.\ {\bf 334} (2015) 977-1023. 

\bibitem{schlue2}
V.\ Schlue \emph{Decay of the Weyl curvature in expanding black hole cosmologies}, arXiv:1610.04172.

\bibitem{Speck:2012}
J.\ Speck, \emph{The nonlinear future-stability of the FLRW family of solutions to the Euler-Einstein
system with a positive cosmological constant}, 	
Sel.\ Math.\ New Ser.\ {\bf 18} (2012) 633-715. 

\bibitem{Tchapnda:2005}
S.\ Tchapnda and N.\ Noutchegueme \emph{The surface-symmetric Einstein-Vlasov
system with cosmological constant}, Math.\ Proc.\ Cambridge Phil.\ Soc.\ {\bf 138} (2005), 
541-553. 

\bibitem{Tchapnda:2003}
S.\ Tchapnda and A.\ Rendall, \emph{Global existence and asymptotic behavior
in the future for the Einstein-Vlasov system with positive cosmological constant}, Class.\ Quant.\
Grav.\ {\bf 20} (2003) 3037–3049. 

\bibitem{Wald:1983}
R.\ Wald, \emph{Asymptotic behavior of homogeneous cosmological models in the presence of
a positive cosmological constant}, Phys.\ Rev.\ D {\bf 28} (1983) 2118–2120.


\end{thebibliography}
\end{document}